\documentclass[a4paper,11pt]{article}

\usepackage[margin=1in]{geometry}  
\usepackage{amsmath, amssymb}      
\usepackage{graphicx}              
\usepackage{booktabs}              
\usepackage{hyperref}              
\usepackage{bm}
\usepackage{bbm}
\usepackage{algorithm}
\usepackage[noend]{algpseudocode}
\usepackage{lipsum}                
\usepackage{multirow}
\usepackage{xcolor}
\usepackage{verbatim}
\usepackage[color=lime]{todonotes}
\usepackage{booktabs,siunitx}
\usepackage{amsmath, amssymb, amsthm, bm}

\theoremstyle{definition}    
\newtheorem{example}{Example}[section]

\newtheorem{theorem}{Theorem}[section]

\usepackage[numbers]{natbib}
\newcommand{\mv}[1]{\boldsymbol{#1}}
\newcommand{\E}{\mathbb{E}}                   
\DeclareMathOperator{\tr}{tr}                 

\newcommand{\uarg}{\;\cdot\;}                 
\title{Adaptive Riemannian Manifold Hamiltonian Monte Carlo with Hierarchical Metric}
\author{Miika Kailas$^1$ \and Matti Vihola$^1$ \and Jonas Wallin$^2$ \\[1ex]
$^1$ Department of Mathematics and Statistics, University of Jyväskylä \\
$^2$ Department of Statistics, Lund University
}
\date{}

\begin{document}

\maketitle

\begin{abstract}
Hamiltonian Monte Carlo (HMC) and its dynamic extensions, such as the No-U-Turn Sampler (NUTS),
 are powerful Markov chain Monte Carlo methods for sampling from complex, high-dimensional probability distributions.
  Riemannian manifold Hamiltonian Monte Carlo (RMHMC) extends HMC by allowing the mass matrix to depend on position, 
  which can substantially improve mixing but also makes implementation considerably more challenging.
   In this paper, we study an adaptive hierarchical version of RMHMC that is well suited to many hierarchical sampling problems.
    A key feature of hierarchical RMHMC is that, unlike general RMHMC, it admits a closed-form explicit leapfrog integrator, enabling efficient implementation
     and direct use within dynamic HMC methods such as NUTS. 
     We introduce an adaptive scheme that automatically tunes the parameters of the hierarchical mass matrix during simulation. 
     Importantly, the target density need not exhibit any hierarchical or block structure;
      the hierarchy is instead imposed on the mass matrix as a modeling device to capture the local geometry of the target distribution. 
      Numerical experiments demonstrate appealing empirical performance in high-dimensional Bayesian inference problems.
  \vspace{1em}

\end{abstract}

\section{Introduction}
\label{sec:intro}

Hamiltonian (or hybrid) Monte Carlo (HMC) \citep{duane1987hybrid,neal2011,Betancourt2017} are  Markov chain Monte Carlo (MCMC) methods designed to sample from a generic target probability density in $\mathbb{R}^d$ having the following form:
\begin{align}
  \label{eq:target_distribution}
  \pi(\mv{\theta}) = \frac{\pi_u(\mv{\theta})}{\int \pi_u(\mv{\theta}') \, d\mv{\theta}'},\qquad
  \text{where}\qquad
  \pi_u(\mv{\theta}) = e^{-U(\mv{\theta})},
\end{align} 
with (potential energy) function $U(\mv{\theta}) = - \log \pi_u(\mv{\theta})$ that is finite and differentiable. To sample with HMC, one only needs access to $U(\mv{\theta})$ and its gradient $\nabla U(\mv{\theta})$. The rationale of HMC is to mimic simulation of (fictitious) Hamiltonian dynamics in a phase space (i.e.~with \emph{positions} and \emph{momenta}), and therefore be able to make large, coherent jumps across the target distribution, avoiding diffusive behavior that many other MCMC algorithms based on local moves can have, particularly in higher dimensions~\cite[cf.][]{roberts-rosenthal-optimal-scaling}.

While HMC methods work efficiently in many scenarios, it is well-known that they can be sensitive to the geometry of the target distribution. The Neal's funnel distribution \citep{neal2003slice} is a canonical example with a bad geometry for HMC methods. It is defined as the following two-level hierarchical model with normal prior and conditionals:
\begin{equation}
\label{eq:funnel_distribution}
  v 
  \sim \mathcal{N}\left(0,3^2\right), 
  \qquad\text{and}\qquad
  \mv{x} 
  \sim \mathcal{N}\left(0, \exp(v)\bf{I}\right).
\end{equation}
The changing scale of the conditional distribution of $\mv{x}$ in terms of $v$ creates the `funnel' like geometry, which causes slow mixing for many widely used HMC algorithms.


There are two main approaches for mitigating the mixing issues and thereby improving the efficiency of HMC: (1) improving the integrator of the HMC, or (2) 
changing the geometry of the target distribution or the Hamiltonian. For instance, \citep{bourabee-carpenter-kleppe-marsden,bou2025within,modi2024delayed} (see also \citep{biron2024automala} regarding a Metropolis-adjusted Langevin algorithm) present integrators with adaptive step sizes which can mitigate the issues of dynamic
HMC with targets having varying scales, such as the funnel distribution. However, even if the integrator would be \emph{ideal}, 
such a target can be difficult to explore for a HMC sampler; we discuss this aspect further in Section \ref{sec:geometry_hierarchical_models}.

Our focus will be on the latter strategy, where the most common method is to reparameterize the target distribution so that it is more amenable to HMC. 
This can be done for instance manually \citep{papaspiliopoulos2007general} or by learning a transformation \citep[e.g.][]{hoffman2019neutra}. 
Non-Gaussian momentum distributions have also been suggested to improve mixing \citep{livingstone2019kinetic}. 
A Gaussian kinetic energy distribution can also be combined with a position-dependent mass matrix; this is an incomplete reparameterization in the terminology of \citet{betancourt2019incomplete}. 
The position-dependent mass matrix leads to the Riemannian manifold HMC (RMHMC) \citep{GirolamiCalderhead2011}, which we discuss further in Section \ref{sec:rmhmc}. 
RMHMC is often limited in practice by its computational cost.

In particular, the generalized leapfrog integrator entails implicit updates, and each step requires solving linear systems involving the mass matrix. 
One way to alleviate this difficulty is to reformulate the dynamics in terms of velocity rather than momentum, which yields a semi-explicit integrator in which the position update is explicit while the velocity update remains implicit; see \citep{lan2015markov} for details of a method admitting a fully explicit integrator.
Another possibility is to impose additional structure on the position-dependent mass matrix. In particular, one can specify a position-dependent diagonal hierarchical mass-matrix structure so that the integrator admits an explicit form \citep{zhang2014semi,kleppe2019dynamically}. 
This hierarchical RMHMC method is presented in Section \ref{sec:hierarchical_rhmhc}. 
We detail how the hierarchical RMHMC can be used within a dynamic HMC such as the No-U-Turn Sampler (NUTS) in Appendix \ref{sec:nuts-trajectory}.
dynamic HMC such as the No-U-Turn Sampler (NUTS) in Appendix \ref{sec:nuts-trajectory}.

Choosing the (position-dependent) mass matrix in the HMC is linked with a broader question of selecting the geometry or `information' matrices of the proposal for MCMC algorithms, or preconditioning, that is, applying a transformation on the target distribution before sampling \citep{hird-livingstone}. 
In the random-walk Metropolis, the Metropolis adjusted Langevin algorithm (MALA), and the HMC, a common strategy is to use the inverse empirical covariance of pilot draws as the preconditioner.
Beyond this, several other information-based metrics have been proposed, some of which we have gathered to Table \ref{tab:info}.
Note that $\mathcal{I}_1$ and $\mathcal{I}_3$ are global, in contrast with $\mathcal{I}_2,\mathcal{I}_4$ and $\mathcal{I}_5$
 which are local, that is, depend on the position $\mv{\theta}$.

\begin{table}[H]
  \caption{Information matrices in the literature. Expectations are taken with respect to (a conditional of) $\pi$.}
  \label{tab:info}
  \medskip
  \begin{center}
\begin{tabular}{lll}
  \toprule
  Name & Refs & Definition \\
  \midrule
Inverse posterior covariance 
& \citep{haario2001} 
& $\mathcal{I}_1 \;=\; \mathrm{Cov}(\pi)^{-1}$ \\
Fisher information metric 
&\citep{GirolamiCalderhead2011}
& $\mathcal{I}_2(\mv{\theta})\;=\; \mathbb{E}_{Y\mid \mv{\theta}}\!\left[-\nabla_{\mv{\theta\theta}}^2 \log \pi(\mv{\theta}\mid Y)\right]
$ \\
Average observed information 
&\citep{titsias2023optimal}
& $
\mathcal{I}_3 \;=\; \mathbb{E}_{\mv{\theta}}\!\left[\nabla_{\mv{\theta\theta}}^2 U (\mv{\theta})\right]
$\\
Local negative Hessian 
& \citep{petra2014computational}
& $
\mathcal{I}_4(\mv{\theta})\;=\;\nabla_{\mv{\theta\theta}}^2U(\mv{\theta})$\\
 Monge metric
& \citep{hartmann2022lagrangian} &
$ \mathcal{I}_{5}(\mv{\theta})
\;=\;
{\bf I} + \alpha^2 \nabla \pi(\mv{\theta})\nabla \pi(\mv{\theta})^\mathsf{T}
$ \\
Average conditional observed information & \citep{kleppe2019dynamically} &
$ \mathcal{I}_{B\mid A}(\mv{\theta}_A)
\;=\;
\mathbb{E}_{\mv{\theta}_B}\!\left[-\nabla_{\mv{\theta}_B\mv{\theta}_B}^2\log \pi(\mv{\theta})\,\middle|\,\mv{\theta}_A\right].
$ \\
\bottomrule \\
\end{tabular}
\end{center}
\end{table}

In \citet{hartmann2022lagrangian}, the authors employ the solver of \citet{lan2015markov}. 
In that setting, efficient integration is possible because the Monge metric depends on $\mv{\theta}$ only through a rank-one modification of the identity. 
In the present work, by contrast, we focus on metrics with a more general high-dimensional dependence.
In many hierarchical models, there is a natural block structure, that is, the posterior variables, or the coordinates of the parameter vector, are allocated to two blocks $\mv{\theta} = (\mv{\theta}_A, \mv{\theta}_B)$---often the `parameters' and the `latent variables', respectively. For the first block $A$ we use the regular average observed information matrix $\left(\mathcal{I}_3 \right)_{AA}$, while for the second block $B$ we want a matrix that depends on $\mv{\theta}_{A}$.
A naive construction, like $\mathcal{I}_2(\mv{\theta})_{BB}$ or $\mathcal{I}_4(\mv{\theta})_{BB}$ for block $B$, would be impractical, because they typically depend on \emph{both} $\mv{\theta}_A$ and $\mv{\theta}_B$, and would therefore be incompatible with our hierarchical RMHMC.
Further are expensive to evaluate (for $\mathcal{I}_2$), or lack positive-definiteness guarantees (for $\mathcal{I}_4$).

In \cite{kleppe2019dynamically}, Kleppe used the conditional observed information matrix:
\begin{equation}
\label{eq:conditional_fisher}
\mathcal{I}_{B\mid A}(\mv{\theta}_A)
\;=\;
\mathbb{E}_{\mv{\theta}_B}\!\left[-\nabla_{\mv{\theta}_B\mv{\theta}_B}^2\log \pi(\mv{\theta})\,\middle|\,\mv{\theta}_A\right],
\end{equation}
in a hierarchical manner, and in fact allowing for multi-block extensions.
Recently, Kleppe \cite{kleppe2024log} proposed approximating the conditional information matrix using a metric tensor that can be done efficiently with automatic differentiation.
In Section \ref{sec:mass_matrix_estimation}, we present our method to adaptively estimate the conditional information matrix.
The overall idea is to formulate a distributional approximation of the joint density of $\mv{\theta}_A$ 
and the score corresponding to $\mv{\theta}_B$, that is, $\mv{g}_B=\nabla_{\mv{\theta}_B}\log \pi(\mv{\theta})$.
The approximation of the gradient is done using a flexible parametric approximation,
which we can estimate by minimizing the Kullback--Leibler divergence between their true joint density and the approximation model.
We focus on the case where $\mathcal{I}_{B\mid A}(\mv{\theta}_A)$ is approximated by a diagonal matrix, but stress that the proposed method gives a general framework where other, more complicated models can be used. 

The Kullback--Leibler minimization is in principle direct to implement in practice using a stochastic gradient approach, while running the (dynamic) RMHMC.
The method therefore belongs to the class of MCMC algorithms which adjust their parameters using past simulated samples, which are known as adaptive MCMC \citep[e.g.][]{haario2001,andrieu2008tutorial,roberts2009examples}. For instance \citep{haario2001,vihola2012robust,wallin2018efficient,atchade2006adaptive}  
adjust the proposal covariance of a random-walk Metropolis or MALA, which is equivalent to a global linear preconditioning, and \citep{bell2024adaptive} 
learn the parameters a stereographic projection. There are also algorithms which learn different parameters for a number of regions in the state space \citep[e.g.][]{craiu2009learn,pompe2020framework}, 
and for different temperature levels in parallel tempering \citep{miasojedow2013adaptive}.

While our adaptation (mass matrix estimation) is in principle of now-standard stochastic gradient type \citep[cf.][]{andrieu2008tutorial}, it is relatively well-known that HMC can be challenging for adaptive MCMC.
For instance, the standard adaptation strategy for random-walk Metropolis based on the estimated target covariance \citep{haario2001,andrieu-moulines} is relatively well-behaved \citep[cf.][]{saksman-vihola,vihola-collapse}, but can behave poorly with HMC (and MALA) (see, e.g. \cite{livingstone2022barker,tran-kleppe} and the concurrent work \cite{seyboldt-carlson-carpenter}).
We have implemented new stabilization mechanisms into the adaptation, which we discuss in Section \ref{sec:full-algorithm} and Appendix \ref{sec:clipping-details}. They can be of their independent interest for developing other types of adaptive HMC algorithms. 
To assess the performance of the proposed hierarchical RMHMC method with adaptive mass-matrix estimation, 
we perform a series of numerical experiments in Section \ref{sec:experiment}, and conclude with a discussion in Section \ref{sec:discussion}.

\section{Riemannian manifold Hamiltonian Monte Carlo}
\label{sec:rmhmc}

The Riemannian manifold HMC (RMHMC) \citep{GirolamiCalderhead2011} is a generalization of the HMC by allowing for a non-constant mass matrix. 
That is, one defines a position-dependent mass matrix $\mv{M}(\mv{\theta})\in\mathbb{R}^{d\times d}$ which aims for adapting to the local geometry of the target density. 
The mass matrices need to be symmetric and positive definite: $\mv{M}(\mv{\theta})\!\succ\! 0$ for all $\mv{\theta}$, and the mapping $\mv{\theta}\mapsto \mv{M}(\mv{\theta})$ must be differentiable. 
The Hamiltonian is defined as follows:
\begin{align}
  H(\mv{\theta},\mv{p})
  \;=\;
  U(\mv{\theta})
  \;+\;
  \tfrac{1}{2}\,\mv{p}^\mathsf{T}\mv{M}(\mv{\theta})^{-1}\mv{p}
  \;+\;
  \tfrac{1}{2}\,\log\!\big|\mv{M}(\mv{\theta})\big|,
  \label{eq:rmhmc_hamiltonian}
\end{align}
which corresponds to the conditional momentum law
$\mv{p}\,|\,\mv{\theta}\sim\mathcal{N}\!\big(\mv{0},\mv{M}(\mv{\theta})\big)$ and the augmented target density
\[
  \pi(\mv{\theta},\mv{p}) \;\propto\; \pi(\mv{\theta})\,\big|\mv{M}(\mv{\theta})\big|^{-1/2}
  \exp\!\Big\{-\tfrac{1}{2}\mv{p}^\mathsf{T}\mv{M}(\mv{\theta})^{-1}\mv{p}\Big\}.
\]
Note that with constant mass $\mv{M}(\mv{\theta}) = \mv{M}$, the last term in \eqref{eq:rmhmc_hamiltonian} becomes redundant, and we recover the usual HMC, so RMHMC is a proper generalization of the HMC.

The \emph{ideal} Hamiltonian dynamics corresponding to Hamiltonian \eqref{eq:rmhmc_hamiltonian} are:
\begin{align}
  \dot{\mv{\theta}} \;=\; \nabla_{\mv{p}} H(\mv{\theta},\mv{p}) \;=\; \mv{M}(\mv{\theta})^{-1}\mv{p},
  \qquad
  \dot{\mv{p}} \;=\; -\nabla_{\mv{\theta}} H(\mv{\theta},\mv{p}),
  \label{eq:rmhmc_flow}
\end{align}
where, writing $\mv{v}=\mv{M}(\mv{\theta})^{-1}\mv{p}$ and $\partial_i\equiv \partial/\partial\theta_i$, the $i$th component of $\dot{\mv{p}}$ is
\begin{align*}
  \dot{p}_i
  \;=\;
  -\,\partial_i U(\mv{\theta})
  \;+\;\tfrac{1}{2}\,\mv{v}^\mathsf{T}\big(\partial_i \mv{M}(\mv{\theta})\big)\mv{v}
  \;-\;\tfrac{1}{2}\,\mathrm{tr}\!\Big(\mv{M}(\mv{\theta})^{-1}\,\partial_i \mv{M}(\mv{\theta})\Big).
\end{align*}
Let $\Phi_\tau$ denote the exact time-$\tau$ Hamiltonian flow map associated with \eqref{eq:rmhmc_flow}. 
If $\Phi_\tau$ could be evaluated exactly, one could define an \emph{exact} RMHMC update that leaves $\pi(\mv{\theta},\mv{p})$ invariant (see \citealp[Theorem~5.1]{bou2018geometric}) by:
draw $\mv{p}^{*} \sim \mathcal{N}\!\big(\mv{0},\mv{M}(\mv{\theta}^{(t)})\big)$ and set
\[
(\mv{\theta}^{(t+1)},\mv{p}^{(t+1)}) \;=\; \Phi_\tau\!\big(\mv{\theta}^{(t)},\mv{p}^{*}\big).
\]

However, in practice the dynamics in \eqref{eq:rmhmc_flow} cannot be solved exactly (see \cite{pakman2014exact} for a notable exception), so one must use numerical integration. 
Given a step size $\epsilon$ and number of steps $L$, the integrator evolves the Hamiltonian dynamics for a total integration time $\tau = L\epsilon$.
A standard choice is the \emph{leapfrog} family of symplectic, time-reversible maps. 
For standard HMC (with a separable Hamiltonian), the leapfrog updates are explicit. 
In RMHMC, however, the position-dependent term typically yields a non-separable Hamiltonian, and the resulting \emph{generalized leapfrog} step is implicit and must be solved numerically.

For the generalized leapfrog, define
\begin{align}
\label{eq:XI}
\Xi(\mv{\theta},\mv{p})
&= -\nabla U(\mv{\theta})
+ \tfrac{1}{2}
\begin{bmatrix}
\mv{v}^\mathsf{T}(\partial_1 \mv{M}(\mv{\theta}))\mv{v} - \mathrm{tr}\!\big(\mv{M}(\mv{\theta})^{-1}\partial_1 \mv{M}(\mv{\theta})\big)\\[-2pt]
\vdots\\[-2pt]
\mv{v}^\mathsf{T}(\partial_d \mv{M}(\mv{\theta}))\mv{v} - \mathrm{tr}\!\big(\mv{M}(\mv{\theta})^{-1}\partial_d \mv{M}(\mv{\theta})\big)
\end{bmatrix}.
\end{align}
Algorithm~\ref{alg:genlf_step} gives pseudo-code for one implicit, symmetric generalized leapfrog step $\Psi_{\epsilon}$.

\begin{algorithm}
\caption{Generalized leapfrog step $\Psi_{\epsilon}$ (implicit, symmetric)}
\label{alg:genlf_step}
\begin{algorithmic}[1]
\Require Current state $(\mv{\theta}^{(n)},\mv{p}^{(n)})$, step size $\epsilon$
\Require Implicit solver $\textsc{Solve}(\cdot)$ 
\State $\mv{p}^{(n+\tfrac12)} \leftarrow
\textsc{Solve}\!\Big(\mv{p}=\mv{p}^{(n)}+\tfrac{\epsilon}{2}\,\Xi(\mv{\theta}^{(n)},\mv{p})\Big)$
\State $\mv{\theta}^{(n+1)} \leftarrow
\textsc{Solve}\!\Big(\mv{\theta}=\mv{\theta}^{(n)}+\tfrac{\epsilon}{2}\big(\mv{v}(\mv{\theta}^{(n)},\mv{p}^{(n+\tfrac12)})
+\mv{v}(\mv{\theta},\mv{p}^{(n+\tfrac12)})\big)\Big)$
\State $\mv{p}^{(n+1)} \leftarrow \mv{p}^{(n+\tfrac12)}
+\tfrac{\epsilon}{2}\,\Xi(\mv{\theta}^{(n+1)},\mv{p}^{(n+\tfrac12)})$
\State \Return $(\mv{\theta}^{(n+1)},\mv{p}^{(n+1)})$
\end{algorithmic}
\end{algorithm}

Because numerical integration introduces discretization error, the endpoint of the simulated trajectory cannot be accepted automatically. 
Instead, one apply a Metropolis--Hastings correction step. 
A basic RMHMC iteration is summarized in Algorithm~\ref{alg:simple_RMHMC_basic}.
\begin{algorithm}
\caption{One iteration of a basic Riemannian Manifold HMC}
\label{alg:simple_RMHMC_basic}
\begin{algorithmic}[1]
\Require Current position $\mv{\theta}^{(t)}$, step size $\epsilon$, number of steps $L$
\State Set $\mv{\theta}_0 \gets \mv{\theta}^{(t)}$ and simulate $\mv{p}_0 \sim \mathcal{N}\!\big(\mv{0}, \mv{M}(\mv{\theta}^{(t)})\big)$ 
\For{$n=1,\dots,L$} \Comment{Integrate dynamics}
  \State $(\mv{\theta}_{n},\mv{p}_{n}) \gets \Psi_{\epsilon}(\mv{\theta}_{n-1},\mv{p}_{n-1})$
\EndFor
\State $\alpha \gets \min\Bigl\{1,\exp\bigl(H(\mv{\theta}_0,\mv{p}_0)-H(\mv{\theta}_L,\mv{p}_L)\bigr)\Bigr\}$ 
\State With probability $\alpha$, set $\mv{\theta}^{(t+1)} \gets \mv{\theta}_L$; otherwise set $\mv{\theta}^{(t+1)} \gets \mv{\theta}_0$
\State \Return $(\mv{\theta}^{(t+1)})$
\end{algorithmic}
\end{algorithm}

Algorithm~\ref{alg:simple_RMHMC_basic} can be extended (and improved) in two directions.
First, instead of considering only the endpoint of the numerical trajectory, one may treat \emph{all} states along the trajectory as candidate proposals and select an index $I\in\{0,\dots,T\}$, for instance, using multinomial sampling:
\[
\mathbb{P}(I=t) \;\propto\; \exp\!\big\{-H(\mv{\theta}_t,\mv{p}_t)\big\}, \qquad t=0,\dots,T.
\]
Second, one can avoid fixing the integration time in advance by building the trajectory dynamically and, as in NUTS \citep{HoffmanGelman2014}, allowing for expanding it both forwards and backwards in (fictitious) time from an initial state $(\mv{\theta}^*,\mv{p}^*)$
 until a stopping criterion is met (e.g.\ until a u-turn condition triggers).
A full iteration of an abstract generalized RMHMC scheme is given in Algorithm~\ref{alg:genl_RMHMC}.

\begin{algorithm}
\caption{One iteration of Riemannian Manifold HMC (RMHMC)}
\label{alg:genl_RMHMC}
\begin{algorithmic}[1]
\Require Current position $\mv{\theta}^{(t)}$, step size $\epsilon$, 
\State Set $\mv{\theta}^* \gets \mv{\theta}^{(t)}$ and simulate $\mv{p}^* \sim \mathcal{N}\!\big(\mv{0}, \mv{M}(\mv{\theta}^{*})\big)$ 
\State $(\mv{\theta}_0,\mv{p}_0),(\mv{\theta}_1,\mv{p}_1),\ldots,(\mv{\theta}_T,\mv{p}_T) \sim \textsc{Trajectory}\big((\mv{\theta}^*,\mv{p}^*),\Psi_{\epsilon}) \big)$
\State $I \sim \textsc{Selection}\big((\mv{\theta}_0,\mv{p}_0),(\mv{\theta}_1,\mv{p}_1),\ldots,(\mv{\theta}_T,\mv{p}_T) \big)$
\State \Return $\mv{\theta}^{(t+1)} = \mv{\theta}_I$
\end{algorithmic}
\end{algorithm}

Reversibility of the resulting Markov transition with respect to the joint target $\pi(\mv{\theta},\mv{p})$ requires the same core ingredients as in HMC:
$\Psi_\epsilon$ should be symmetric (time-reversible) and volume-preserving, and the trajectory-building and selection rule must be constructed to satisfy detailed balance for $\pi(\mv{\theta},\mv{p})$; see for instance \citep{durmus-gruffaz-kailas-saksman-vihola}.
The NUTS-like RMHMC variant which we use in the experiments is detailed in Appendix~\ref{sec:nuts-trajectory}.

While RMHMC is an appealing generalization of HMC, the implicit solver in Algorithm~\ref{alg:genlf_step} 
makes it substantially more expensive than standard HMC.
Although the Metropolis--Hastings correction ensures validity of the method for any step and trajectory length parameters $(\epsilon,L)$, the implicit solver tolerance must be small: an inaccurate solver can break the symmetry/reversibility of $\Psi_\epsilon$ and may thereby introduce a bias. Finally, nonlinear solvers may impose additional stability constraints that further restrict the usable step size $\epsilon$, to avoid convergence failures.

\section{Hierarchical Riemannian manifold HMC}
\label{sec:hierarchical_rhmhc}

We consider the special case of RMHMC where the mass matrix $\mv{\theta}$ admits a specific hierarchical structure, which stems naturally from hierarchical target distributions. 
For the sake of exposition, we first focus on the case of two blocks: an upper hierarchical block of variables $\mv{\theta}_A$ that controls the local scale/curvature of the remaining block $\mv{\theta}_B$. 
The funnel \eqref{eq:funnel_distribution} is a prototypical example of this, with $\mv{\theta}_A=v$ and $\mv{\theta}_B=\mv{x}$.

\paragraph{General two-block hierarchical mass matrix}

By exploiting the structure, one can use a block-diagonal mass matrix  \cite{zhang2014semi,kleppe2019dynamically,kleppe2024log}:
\begin{equation}
  \mv{M}(\mv{\theta}) \;=\;
  \begin{bmatrix}
    \mv{M}_A & \mv{0} \\
    \mv{0}   & \mv{M}_B(\mv{\theta}_A)
  \end{bmatrix},
  \qquad
  \mv{M}_A \succ 0,\;\; \mv{M}_B(\mv{\theta}_A)\succ 0.
  \label{eq:hierarchical_metric}
\end{equation}
In words, $\mv{\theta}_A$ evolves in a constant (Euclidean) geometry while $\mv{\theta}_B$ adapts through $\mv{M}_B(\mv{\theta}_A)$. 
The Hamiltonian is identical to the regular RMHMC \eqref{eq:rmhmc_hamiltonian}, but thanks to the block structure, the generalized Hamiltonian  admits an explicit symmetric splitting integrator.

To that end, using the hierarchical structure, we can write $\mv{v} = [\mv{v}_A,\mv{v}_B]=[\mv{M}_A^{-1}\mv{p}_A,\mv{M}_B(\mv{\theta}_A)^{-1}\mv{p}_B]$, and similarly decompose $\Xi(\mv{\theta},\mv{p})$ from \eqref{eq:XI} into two parts $\Xi(\mv{\theta},\mv{p}) = [\Xi_A(\mv{\theta},\mv{p}_B)^\top, \Xi_B(\mv{\theta})^\top]^\top$.
 Since
$\partial_{\mv{\theta}_i}\mv{M}(\mv{\theta})=0$ for all $i\in B$, we have $ \Xi_B(\mv{\theta}) = -\nabla_{\mv{\theta}_B}U(\mv{\theta})$.
Let $A_i$ denote the $i$th element of of block $A$, then
\begin{align}
  \left[\Xi_A(\mv{\theta},\mv{p}_B) \right]_i
  &= -\nabla_{\mv{\theta}_{A_i}}U(\mv{\theta})
  + \tfrac{1}{2}\Big(\mv{v}_B^\top \left[\partial_{\mv{\theta}_{A_i}}\mv{M}_B(\mv{\theta}_A) \right]\mv{v}_B
  -\tr\!\big[\mv{M}_B(\mv{\theta}_A)^{-1}\partial_{\mv{\theta}_{A_i}}\mv{M}_B(\mv{\theta}_A)\big]\Big).
  \label{eq:xi_A_B_repeat}
\end{align}
This removes the implicit coupling in the generalized leapfrog and yields an explicit, symmetric and volume-preserving update, which is given in Algorithm~\ref{alg:hier_lf_step}. 
\begin{algorithm}
\caption{Hierarchical generalized leapfrog step $\Psi_{\epsilon}$ (explicit)}
\label{alg:hier_lf_step}
\begin{algorithmic}[1]
\Require Current state $(\mv{\theta}^{(n)},\mv{p}^{(n)})$, step size $\epsilon$
\State $\mv{p}_B^{(n+\tfrac12)} \leftarrow \mv{p}_B^{(n)} + \tfrac{\epsilon}{2}\,\Xi_B(\mv{\theta}_A^{(n)},\mv{\theta}_B^{(n)})$
\State $\mv{p}_A^{(n+\tfrac12)} \leftarrow \mv{p}_A^{(n)} + \tfrac{\epsilon}{2}\,\Xi_A(\mv{\theta}_A^{(n)},\mv{\theta}_B^{(n)},\mv{p}_B^{(n+\tfrac12)})$
\State $\mv{\theta}_A^{(n+1)} \leftarrow \mv{\theta}_A^{(n)} + \epsilon\,\mv{M}_A^{-1}\mv{p}_A^{(n+\tfrac12)}$
\State $\mv{\theta}_B^{(n+1)} \leftarrow \mv{\theta}_B^{(n)} + \tfrac{\epsilon}{2}\Big(
  \mv{M}_B(\mv{\theta}_A^{(n)})^{-1} + \mv{M}_B(\mv{\theta}_A^{(n+1)})^{-1}\Big)\mv{p}_B^{(n+\tfrac12)}$
\State $\mv{p}_B^{(n+1)} \leftarrow \mv{p}_B^{(n+\tfrac12)} + \tfrac{\epsilon}{2}\,\Xi_B(\mv{\theta}_A^{(n+1)},\mv{\theta}_B^{(n+1)})$
\State $\mv{p}_A^{(n+1)} \leftarrow \mv{p}_A^{(n+\tfrac12)} + \tfrac{\epsilon}{2}\,\Xi_A(\mv{\theta}_A^{(n+1)},\mv{\theta}_B^{(n+1)},\mv{p}_B^{(n+\tfrac12)})$
\State \Return $(\mv{\theta}^{(n+1)},\mv{p}^{(n+1)})$
\end{algorithmic}
\end{algorithm}
The following result ensures that integrator defined by Algorithm \ref{alg:hier_lf_step} has the required properties so that it can be used within the RMHMC Algorithm \ref{alg:genl_RMHMC}, while keeping the RMHMC step remains $\pi$-reversible:
\begin{theorem}
\label{thm:symplectic_reversible_hier_leapfrog}
The $\Psi_\epsilon$ defined by Algorithm \ref{alg:hier_lf_step} is (i) symmetric (time-reversible) and (ii) volume-preserving.
\end{theorem}
In Appendix \ref{sec:symmetry_verification} we prove an extension of Theorem \ref{thm:symplectic_reversible_hier_leapfrog} which allows for a more general block structure.

\paragraph{Diagonal mass for block B}

If the block B has a high dimension, calculating the matrix inverse $\mv{M}_B(\mv{\theta}_A)^{-1}$ in Algorithm \ref{alg:hier_lf_step} can be costly. For this reason, we will focus below on the special case where $\mv{M}_B(\mv{\theta}_A)$ is diagonal. In this case, the matrix inverses are direct to calculate, and the term $\Xi_A(\mv{\theta},\mv{p}_B)$ simplifies further. 
Let $M_i(\mv{\theta}_A)$ stand for the diagonal elements of $\mv{M}_B(\mv{\theta}_A)$, 
then one obtains:
\begin{align*}
\mv{v}_B^\top(\partial_{\mv{\theta}_{A_i}}\mv{M}_B(\mv{\theta}_{A}))\mv{v}_B
&= \sum_{j\in B} v_j^2\,\partial_{\mv{\theta}_{A_i}} M_j(\mv{\theta}_A),
\\
\tr\!\big(\mv{M}_B^{-1}(\mv{\theta}_A)\partial_{\mv{\theta}_{A_i}}\mv{M}_B(\mv{\theta}_A)\big)
&= \sum_{j\in B} M_j(\mv{\theta}_A)^{-1}\,\partial_{\mv{\theta}_{A_i}} M_j(\mv{\theta}_A).
\end{align*}
Using these, one can simplify the second term on the right of \eqref{eq:xi_A_B_repeat} with:
\begin{equation}
\label{eq:nuA_diag}
\tfrac{1}{2}\sum_{j\in B}\Big(v_j^2 - M_j(\mv{\theta}_A)^{-1}\Big)\,\nabla_{\mv{\theta}_{A_i}} M_j(\mv{\theta}_A)
\;=\;
\tfrac{1}{2}\sum_{j\in B}\Big(\tfrac{p_{B,j}^2}{M_j(\mv{\theta}_A)} - 1\Big)\,\nabla_{\mv{\theta}_{A_i}}\log M_j(\mv{\theta}_A).
\end{equation}

\paragraph{Multi-block extensions}

Some models may involve additional levels of hierarchy, which can be used to devise multi-block structures. The two block version of the hierarchical RMHMC presented above generalzed into multiple blocks. A sequential multi-block structure was studied by \citet{kleppe2019dynamically}.
In Appendix~\ref{sec:symmetry_verification}, we discuss a slightly more general multi-block extension which need not be sequential, and show that the corresponding integrator remains symmetric and volume-preserving.

\section{Mass matrix estimation}
\label{sec:mass_matrix_estimation}

The hierarchical mass matrix of the form \eqref{eq:hierarchical_metric} admits an efficient RMHMC implementation, but how to select the mass matrices in practice? Our approach to mass-matrix estimation is based on score vectors $\mv{g}(\mv{\theta}) = \nabla_{\mv{\theta}} U(\mv{\theta})$.
If $\mv{\theta}\sim\pi$, it is well-known that the score vectors 
$\mv{g} = \mv{g}(\mv{\theta})$ satisfy (under general conditions)
\begin{equation}
\mathbb{E}[\mv{g}] = \mv{0} \qquad\text{and}\qquad
\mathbb{E}[\mv{g}\mv{g}^\top] = \mathcal{I}_3, 
\label{eq:gradient_mean_general}
\end{equation}
where $\mathcal{I}_3$ is the average observed information matrix defined in Table~\ref{tab:info}. 
Titsias et al.~\cite{titsias2023optimal} estimate $\mathcal{I}_3$ using the empirical second moment
$\hat{\mathcal{I}}_3=\tfrac{1}{N}\sum_{n=1}^N \mv{g}^{(n)}(\mv{g}^{(n)})^\top$ with $\mv{g}^{(n)} = \mv{g}(\mv{\theta}^{(n)})$ obtained along an MCMC samples $\mv{\theta}^{(n)}$, and Hird and Livingstone \cite{hird-livingstone} show that $\mathcal{I}_3$ can be a good global preconditioner.  However, we are interested in position-dependent mass matrices, which cannot rely on one globally learned quantity.

\subsection{Hierarchical choice of mass matrix}
\label{subsec:mass_matrix_hier_choice}

Recall that we are interested in a mass matrix $\mv{M}$ of the form \eqref{eq:hierarchical_metric}, having block diagonal components $\mv{M}_A$ and, in particular, $\mv{M}_B(\mv{\theta}_A)$, which depends on the position. Namely, let us write the target in terms of its conditionals: $\pi(\mv{\theta}) = \pi_A(\mv{\theta}_A) \pi_{B\mid A}(\mv{\theta}_B\mid \mv{\theta}_A)$. Then, it is easy to see that
$$
\mv{g}_B(\mv{\theta})
=\nabla_{\mv{\theta}_B}\log\pi(\mv{\theta})
= \nabla_{\mv{\theta}_B} \log \pi_{B\mid A}(\mv{\theta}_B\mid \mv{\theta}_A).
$$
That is, for each fixed $\mv{\theta}_A$, the vector $\mv{g}_B$ is a score of the conditional distribution, and therefore, 
under $\mv{\theta}\sim\pi$ it satisfies
\[
\mathbb{E}[\mv{g}_B\mid \mv{\theta}_A]=\mv{0},
\qquad
\mathbb{E}[\mv{g}_B \mv{g}_B^\top\mid \mv{\theta}_A]=\mathcal{I}_{B\mid A}(\mv{\theta}_A),
\]
where $\mathcal{I}_{B\mid A}(\mv{\theta}_A)$ is the average conditional Fisher information \eqref{eq:conditional_fisher}.

We consider a parametric approximation of the conditional score $\mv{g}_B$ given $\mv{\theta}_A$. Assume that $\mv{g}_B$ has conditional density $p_{B\mid A}(\uarg \mid \mv{\theta}_A)$.\footnote{For the distribution to be continuous, it is sufficient that the function $\mv{g}_B(\uarg\mid \mv{\theta}_A)$ is continuously differentiable and that $\det(\nabla_{\mv{\theta}_B} \mv{g}_B(\mv{\theta}_B\mid \mv{\theta}_A)) \neq 0$ almost everywhere. Note, however, that a continuous conditional distribution is not required by our \emph{method}, but only to ensure that it coincides with Kullback--Leibler minimization.} We fit a parametric approximate model $q_{B\mid A}$ which minimizes the overall Kullback--Leibler (KL) divergence:
\begin{equation}
  \mathrm{KL}\big(\pi_A(\mv{\theta}_A) p_{B\mid A}(\mv{g}_B\mid\mv{\theta}_A)\;\|\; \pi_A(\mv{\theta}_A) q_{B\mid A}(\mv{g}_B\mid \mv{\theta}_A)\big).
  \label{eq:kl-conditional-score}
\end{equation}

In particular, we consider approximate models $q_{B\mid A}(\uarg\mid \mv{\theta}_A)$ which correspond to a normal distribution $\mathcal{N}(\mv{0}, \mv{M}_B(\mv{\theta}_A))$ with conditional variance $\mv{M}_B(\mv{\theta}_A; \mv{\phi})\succ 0$ where $\mv{\phi}$ are the parameter for the variance. In this case, 
minimizing \eqref{eq:kl-conditional-score} is equivalent to minimizing the following expected loss (see Appendix \ref{app:conditional-kl}):
\begin{equation}
\E\big[\ell\big(\mv{\phi};\mv{\theta}_A,\mv{g}_B\big)\big], \quad\text{where}\quad
\ell\big(\mv{\phi};\mv{\theta}_A,\mv{g}_B\big)=\log\left|\mv{M}_B(\mv{\theta}_A;\mv{\phi})\right|
+ \mv{g}_B^\top \mv{M}_B(\mv{\theta}_A;\mv{\phi})^{-1} \mv{g}_B,
\label{eq:per_iter_loss_general}
\end{equation}
and where the expectation is with respect to $(\mv{\theta}_A, \mv{g}_B) \sim \pi_A(\mv{\theta}_A) p_{B\mid A}(\mv{g}_B\mid\mv{\theta}_A)$. 

The approximate minimizer of this loss can be found iteratively by a stochastic gradient method:
\begin{equation}
\mv{\phi}^{(n)}
=\mv{\phi}^{(n-1)}
-\eta_n\nabla\ell\big(\mv{\phi}^{(n-1)};\mv{\theta}_A^{(n)},\mv{g}_B^{(n)}\big),
\label{eq:sgd_update_general}
\end{equation}
where $\mv{\theta}_A^{(n)}$ and $\mv{g}_B^{(n)}$ are samples from $\pi_A(\mv{\theta}_A) p_{B\mid A}(\mv{g}_B\mid\mv{\theta}_A)$. As we do not have access to such perfect samples, we use samples from the same MCMC whose parameters we are adjusting. This renders our approach as an instance of Markovian stochastic approximation \citep[cf.][]{benveniste2012adaptive}, and in particular an adaptive MCMC with stochastic approximation dynamic \citep[e.g.][and references therein]{andrieu2008tutorial}.

For the position independent mass of the first block $\mv{M}_A$, we employ a similar approach as for the conditional above, aiming for minimizing $\mathrm{KL}\big(p_{A}(\mv{g}_A)\;\|\; q_{A}(\mv{g}_A)\big)$, where $q_A$ is a Gaussian, leading to a learning rule analogous to \eqref{eq:sgd_update_general}.

Finally, we note that the extension of the learning method discussed above extends directly 
to the multi-block case with arbitrary dependencies as discussed in the end of Section \ref{sec:hierarchical_rhmhc} in an obvious manner. Namely, the `block B' is replaced by `current block' and `block A' by `all other blocks'.

\subsection{Diagonal mass models}

We focus on the diagonal case, where the conditional mass has the form: 
\[
\mv{M}_B(\mv{\theta}_A;\mv{\phi})=\mathrm{diag}\{M_i(\mv{\theta}_A;\mv{\phi}_i): i\in B\},
\qquad M_i(\mv{\theta}_A;\mv{\phi}_i)>0,
\]
where $\mv{\phi}_i$ stand for parameters for each coordinate $i$ of the block $B$.
The diagonal form leads to 
the following per-iteration loss: 
\[
\ell\!\big(\mv{\phi};\mv{\theta}_A,g_B\big)
=\sum_{i\in B}\Big\{\log M_i(\mv{\theta}_A;\mv{\phi}_i)+\frac{g_{B,i}^2}{M_i(\mv{\theta}_A;\mv{\phi}_i)}\Big\},
\]
where ${g_B}_i$ are the coordinates of $\mv{g}_B$. The gradient of this loss with respect to $\mv{\phi}_i$ is:
\[
\nabla_{\mv{\phi}_i}\ell\!\big(\mv{\phi};\mv{\theta}_A,g_B\big)
=\Big(1-\frac{g_{B,i}^2}{M_i(\mv{\theta}_A;\mv{\phi}_i)}\Big)\,\nabla_{\mv{\phi}_i}\log M_i(\mv{\theta}_A;\mv{\phi}_i).
\]

We will present next the two specific parametric models considered further in this work.
In both cases, the correction needed by the integrator is obtained by substituting
$\nabla_{\mv{\theta}_A}\log M_i(\mv{\theta}_A)$ into the diagonal simplification \eqref{eq:nuA_diag} and is given explicit and 
computationally efficient form.

\paragraph{Block exponential model.}
The first model is a simple exponential form of the mass matrix:
\begin{equation}
\label{eq:exponential_model}
M_i(\mv{\theta}_A)=\exp\!\big(\mv{\phi}_i^\top x_i(\mv{\theta}_A)\big),
\qquad
\nabla_{\mv{\phi}_i}\log M_i(\mv{\theta}_A)=x_i(\mv{\theta}_A).
\end{equation}
In this case, $\nabla_{\mv{\theta}_A}\log M_i(\mv{\theta}_A)=\nabla_{\mv{\theta}_A}\big(\mv{\phi}_i^\top x_i(\mv{\theta}_A)\big)$, so
\eqref{eq:nuA_diag} becomes
\[
\nu_A(\mv{\theta}_A,\mv{p}_B)
=\tfrac{1}{2}\sum_{i\in B}\Big(\tfrac{p_{B,i}^2}{M_i(\mv{\theta}_A)}-1\Big)\,
\nabla_{\mv{\theta}_A}\!\big(\mv{\phi}_i^\top x_i(\mv{\theta}_A)\big).
\]
Therefore,  providing the the dimension of $\phi_i$ is $|\phi|$ for all $i$, the cost of computing $\nu_A(\mv{\theta}_A,\mv{p}_B)$ is $\mathcal{O}(|\phi|\,|B|)$.
\begin{example}[Funnel]
\label{ex:funnel_model}
For the funnel model \eqref{eq:funnel_distribution}, choosing $x_i(\mv{\theta}_A)=[1,v]^\top$ and $\mv{\phi}_i=[0,-1]^\top$ gives
$M_i(\mv{\theta}_A)=\exp(-v)=\mathcal{I}_{B\mid A,ii}(\mv{\theta}_A)$, that is, the diagonal oracle is contained in the model class. And we can expect the estimation procedure to recover this solution. 
\end{example}

\paragraph{Block sum-of-exponentials model.}
Our second model is a bit more flexible allowing for each diagonal element of the mass matrix to be a sum of two exponentials.
This is intended to encode masses arising from two components: one due to the prior part and the other for the likelihood.
\begin{equation}
\label{eq:sum_of_exponentials}
M_i(\mv{\theta}_A)
=\exp\!\big((\mv{\phi}^1_i)^\top x_{i,1}(\mv{\theta}_A)\big)
+\exp\!\big((\mv{\phi}^2_i)^\top x_{i,2}(\mv{\theta}_A)\big).
\end{equation}
Let $\eta_{i,k}(\mv{\theta}_A)=(\mv{\phi}^k_i)^\top x_{i,k}(\mv{\theta}_A)$ and
$w_{i,k}(\mv{\theta}_A)=\exp(\eta_{i,k}(\mv{\theta}_A))/M_i(\mv{\theta}_A)$ for $k\in\{1,2\}$.
Then
\[
\nabla_{\mv{\theta}_A}\log M_i(\mv{\theta}_A)
=w_{i,1}(\mv{\theta}_A)\,\nabla_{\mv{\theta}_A}\eta_{i,1}(\mv{\theta}_A)
+w_{i,2}(\mv{\theta}_A)\,\nabla_{\mv{\theta}_A}\eta_{i,2}(\mv{\theta}_A),
\]
and substituting into \eqref{eq:nuA_diag} yields the corresponding $\nu_A(\mv{\theta}_A,\mv{p}_B)$.
Provided that the vectors $\mv{\phi}_i^1$ and $\mv{\phi}_i^2$ have common dimension $|\phi|$ for all $i$, the cost of computing $\nu_A(\mv{\theta}_A,\mv{p}_B)$ remains $\mathcal{O}(|\phi|\,|B|)$.

\section{On the geometry of hierarchical models}
\label{sec:geometry_hierarchical_models}

As discussed in the introduction, hierarchical models often induce funnel-like structures, which are difficult for standard MCMC samplers. The non-centered parameterizations introduced in \citep{papaspiliopoulos2007general} can be used to mitigate these issues. While \citep{papaspiliopoulos2007general} focused on Gibbs samplers, these can be useful more generally. 
In Section \ref{sec:centered-and-noncentered}, we first examine the non-centered parameterization on the funnel example in the context of HMC.
Then, in Section \ref{sec:transformed-hmc}, we discuss following \cite{betancourt2019incomplete} how an appropriate transformation of the mass matrix induces Hamiltonian dynamics that are equivalent to those obtained by the non-centered reparameterization. 

\subsection{Centered and non-centered parameterization}
\label{sec:centered-and-noncentered}

Recall the Neal's funnel distribution \eqref{eq:funnel_distribution}, whose potential function is (up to an additive constant):
$$
U_C(v,\mv{x})=\tfrac12 e^{-v}\|\mv{x}\|^2 + \tfrac{d}{2}v+\tfrac{v^2}{18}.
$$
In the terminology of \citep{papaspiliopoulos2007general}, this is the centered parameterization (C) \(\mv{\theta}=(v,\mv{x})\) of the distribution.
Consider then the non-centered (NC) parameterization \(\tilde{\mv{\theta}}=(\tilde v,\tilde{\mv{x}})\) defined by \(\tilde v=v\), \(\tilde{\mv{x}}=e^{-v/2}\mv{x}\). That is, $\tilde{\mv{\theta}} = T(\mv{\theta})$ with $T^{-1}(\tilde{v}, \tilde{\mv{x}}) = (\tilde{v}, e^{\tilde{v}/2} \tilde{\mv{x}})$. Up to an additive constant, the NC potential is
\[
U_{NC}(\tilde v,\tilde{\mv{x}})
= U_C\big(T^{-1}(\tilde{v},\tilde{\mv{x}})\big) - \log \det \big(D T^{-1}(\tilde{v},\tilde{\mv{x}})\big)
= \tfrac12\|\tilde{\mv{x}}\|^2+\tfrac{\tilde v^{\,2}}{18},
\]
because the determinant of the (lower triangular) Jacobian $\det \big(DT^{-1}(\tilde{v},\tilde{\mv{x}})\big) = e^{d\tilde{v}/2}$
cancels the \((d/2)v\) term in the centered parameterization, and therefore removes the  linear `tilt' in \(v\). 

With a diagonal mass matrix \(\mv{M}=\mathrm{diag}(m_v,m_x\mv{I}_d)\) and Gaussian momenta \(p\sim\mathcal N(\mv{0},\mv{M})\), the kinetic energy  $K=\tfrac12 p^\top \mv{M}^{-1}p$ satisfies \(2K\sim\chi^2_{d+1}\) regardless of the values \((m_v,m_x)\). The typical values concentrate around the expected value \(\E[K]=(d+1)/2\) with typical fluctuations of order \(\mathrm{sd}(K)=\sqrt{(d+1)/2}\). 
This suggests an `energy budget' for a HMC trajectory: 
a typical trajectory starts near total energy \(H\approx U+\E[K]\) with fluctuations of order \(\mathrm{sd}(K)\), 
and any move that would increase \(\Delta U\) far beyond this budget is unlikely. In particular, let us focus on the typical energy changes in terms of \(v\), by considering the conditional expectations of the energies:
\[
U_C(v) = \E[U_C(v,\mv{X})\mid v]=\tfrac{d}{2}+\tfrac{d}{2}v+\tfrac{v^2}{18},
\qquad
U_{NC}(v)= \E[U_{NC}(v,\tilde{\mv{X}})\mid v]=\tfrac{d}{2}+\tfrac{v^2}{18},
\]
because \(\E[\|\mv{X}\|^2\mid v]=de^{v}\).

Note that the expected energy of the centered parameterization grows linearly in \(v\) while in the non-centered case the energy purely quadratic around \(v=0\). 
Figure~\ref{fig:funnel_energy_budget} shows $\Delta U_C(v)$ and $\Delta U_{NC}(v)$ for $d=4$, together with dashed lines indicating the mean and 95\% range of the kinetic energy $K$.
\begin{figure}[t]
  \centering
  \includegraphics[width=0.8\linewidth]{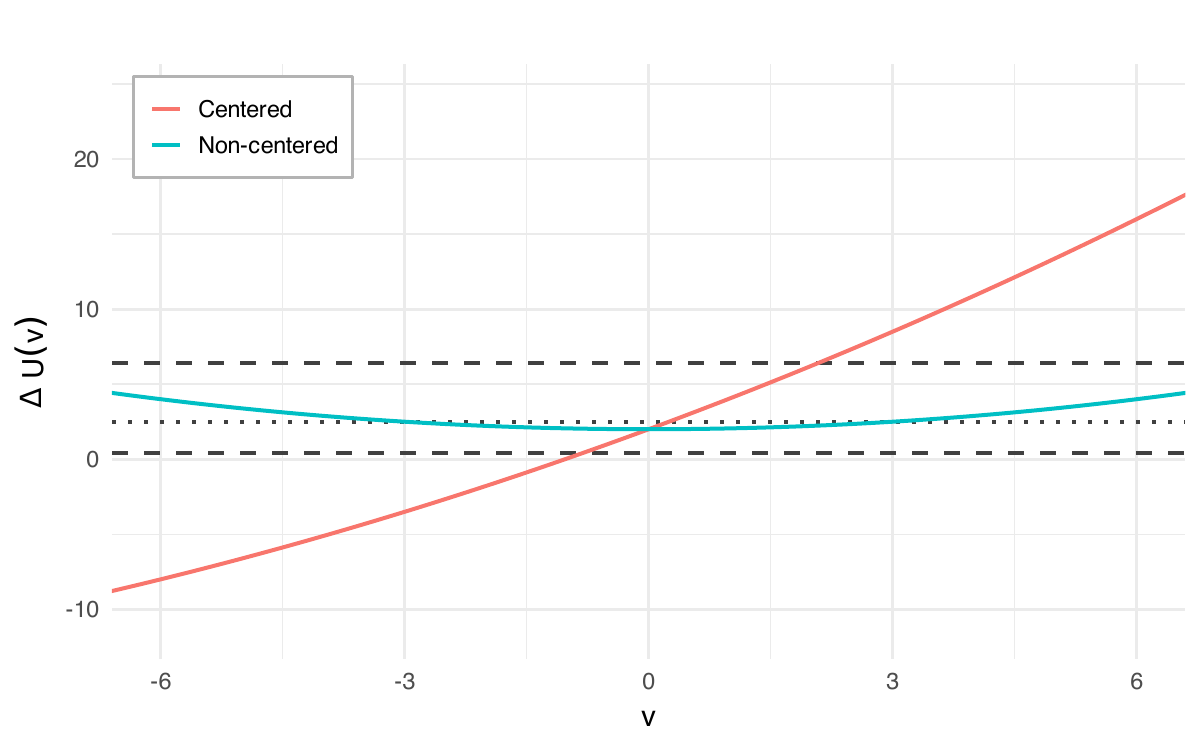}
 \caption{Excess expected potential $\Delta U(v)$ for centered and non-centered forms of the funnel distribution \eqref{eq:funnel_distribution} with $d=4$. The dashed lines indicate the mean and 95\% range of $K$ with $2K\sim\chi^2_{5}$. Regions where $|\Delta U(v)|$ are large are hard to reach in one HMC move. }
 \label{fig:funnel_energy_budget}
\end{figure}
Moving beyond the shown band in one HMC step requires a tail kinetic draw.
The key insight from this example is that the centered parameterization can induce a geometry which is difficult to sample even with \emph{ideal} HMC, that is, having a perfect integrator, and this can happen in a moderate dimension.
This motivates the need for geometry-altering transformations, like a non-centered reparameterization. 

\subsection{Connection to the mass matrix and invariance under reparameterization}
\label{sec:transformed-hmc}

It is possible to apply a smooth transformation (a change of variables) to induce a modified mass matrix which preserves the Hamiltonian dynamics \citet{betancourt2019incomplete}.
Let $\tilde{\mv{\theta}} = T(\mv{\theta})$ be a smooth, invertible transformation
with Jacobian $\mv{J}(\mv{\theta}) = DT(\mv{\theta})$.
To preserve the Hamiltonian dynamics under the change of variables,
the momentum must be transformed as $
\tilde{\mv{p}} = \mv{J}(\mv{\theta})^{-\top} \mv{p}$.
Consider the Hamiltonian $
H(\mv{\theta},\mv{p})
=
U(\mv{\theta})
+
\frac12 \mv{p}^\top \mv{M}(\mv{\theta})^{-1} \mv{p} .
$
Under the transformation, the potential becomes
$
\tilde U(\tilde{\mv{\theta}})
=
U(T^{-1}(\tilde{\mv{\theta}}))
-
\log\!\left|\det \mv{J}(T^{-1}(\tilde{\mv{\theta}}))\right|,
$
and invariance of the kinetic energy implies
$$
\widetilde{\mv{M}}(\tilde{\mv{\theta}})
=
\mv{J}(\mv{\theta})^{-\top}
\mv{M}(\mv{\theta})
\mv{J}(\mv{\theta})^{-1},
\qquad
\mv{\theta}=T^{-1}(\tilde{\mv{\theta}}).
$$
Therefore, reparameterizing the model is equivalent to keeping the
original coordinates but modifying the mass matrix accordingly.
In this sense, a non-centered parameterization can be viewed as an
implicit change of mass matrix.
Apart from the numerical integration error, the resulting Hamiltonian
dynamics describe the same target distribution
\citep{betancourt2019incomplete}.

We can use this result to understand difference between centered and non-centered parameterizations in the funnel example,
by looking at the mass matrix induced by the non-centered reparameterization when viewed in the original coordinates.
\begin{example}
The non-centered map $T$ discussed in Section \ref{sec:centered-and-noncentered} has the following Jacobian:
\[
\mv{J}(v,\mv{x})\;=\;
\begin{bmatrix}
1 & \mv{0}\\[4pt]
-\tfrac12 e^{-v/2}\mv{x} & e^{-v/2}\mv{I}_d
\end{bmatrix}.
\]
Suppose we run HMC in the non-centered parameterization with the simplest choice \(\tilde{\mv{M}}(\tilde{\mv{\theta}})\equiv \mv{I}_d\). 
The mass that makes the \emph{centered} parameterization equivalent is given by:
\begin{equation}
\mv{M}(v,\mv{x})=\mv{J}(v,\mv{x})^\top\tilde{\mv{M}}\mv{J}(v,\mv{x})=
\begin{bmatrix}
1+\tfrac14 e^{-v}\|\mv{x}\|^2 & -\tfrac12 e^{-v}\mv{x}^\top\\[4pt]
-\tfrac12 e^{-v}\mv{x} & e^{-v}\mv{I}_d
\end{bmatrix}.
\label{eq:reparam-matrix}
\end{equation}
\end{example}

In this article we focus on block diagonal mass matrices of the form \eqref{eq:hierarchical_metric}, which does not accommodate mass of the form \eqref{eq:reparam-matrix}, but suggests a block diagonal approximation of it. 
Indeed, consider the common diagonal mass in centered coordinates $(v,\mv{x})$:
\begin{equation}
\mv{M}(v,\mv{x})=\begin{bmatrix}
  M_A     & \mv{0} \\[2pt]
  \mv{0} & \mv{M}_B(v)  \mv{I}_d\end{bmatrix} \quad\text{where}\quad M_A>0\text{ and }\mv{M}_B(v) = e^{-\phi v} \mv{I}_d.
  \label{eq:block-mass}
\end{equation}
This may be viewed as an approximation of \eqref{eq:reparam-matrix}, enforcing off-diagonal blocks zero, and by replacing the first part by a constant (average value of the mass).

Finally, we note that \eqref{eq:block-mass} cannot correspond to a reparameterization. Namely, if this were $\mv{J}^{\!\top}\mv{J}$ for some $T$, a natural candidate \emph{point-wise} is the diagonal field $\mv{J}(v,\mv{x})=\mathrm{diag}\big(1,e^{-v/2}\mv{I}_d\big)$.
But $\mv{J}$ cannot be a Jacobian: the mixed partials do not commute since
$\partial_v(\partial_{x_i}\tilde x_i)=\partial_v(e^{-v/2})=-\tfrac12 e^{-v/2}\neq 0=\partial_{x_i}(\partial_v\tilde x_i)$.
More fundamentally, the metric induced by $\mv{M}$ has \emph{negative curvature}, so it cannot be the pullback of a flat (Euclidean) metric.

\section{Full algorithm and stabilization mechanisms}
\label{sec:full-algorithm}

Before the experiments, we outline the full adaptive RMHMC algorithm in Algorithm \ref{alg:method_overview}.
Lines \ref{line:b-block-1}--\ref{line:b-block-2} implement the estimation of the hierarchical component $\mv{M}_B(\mv{\theta}_A)$, and lines \ref{line:a-block-1}--\ref{line:a-block-2} implement the estimation of the mass matrix $\mv{M}_A$ similarly, with the exception that it is not allowed to depend on $\mv{\theta}$ at all.
Line \ref{line:clipping-mean} implements two stabilization mechanisms: mean adaptation and gradient clipping. 
\begin{algorithm}[t]
\caption{Overview of the method (hierarchical RMHMC metric adaptation)}
\label{alg:method_overview}
\begin{algorithmic}[1]
\Require Target $\pi(\mv{\theta})$; block split $\mv{\theta}=(\mv{\theta}_A,\mv{\theta}_B)$; variance models $\{M_i(\mv{\theta}_A;\phi_i)\}$ (e.g.\ \eqref{eq:sum_of_exponentials}); step size $\epsilon$; iterations $N$; learning rates $\{\eta_k\}$; clipping level $C_{\mv{g}}$.
\State Initialize $\mv{\theta}^{(0)}$, parameters $\mv{\phi}^B=\{\mv{\phi}^B_i\}$, $\mv{\phi}^A=\{\phi^A_i\}$, running mean gradient $\bar{\mv{g}}\gets \mv{0}$.
\For{$k=1,\dots,N$}
  \State \textbf{Metric:} $\mv{M}_B(\theta_A)\gets \mathrm{diag}\!\big(M_i(\mv{\theta}_A;\mv{\phi}^B_i)\big)$, $\mv{M}_A\gets \mathrm{diag}\!\big(M_i(\phi^A_i)\big)$
  \State  $\mv{M}(\theta)\gets \mathrm{diag}\!\big(\mv{M}_A,\;\mv{M}_B(\mv{\theta}_A)\big)$ \hfill (\ref{eq:hierarchical_metric})
  \State \label{item:advance-rmhmc} \textbf{Advance:} $\mv{\theta}^{(k)} \gets \mathrm{RMHMCStep}\!\big(\mv{\theta}^{(k-1)}; \mv{M}, \epsilon\big)$
     \State \label{line:clipping-mean} $\mv{g} \gets \nabla_{\mv{\theta}}\log\pi(\mv{\theta}^{(k)})$;\quad 
     $\bar{\mv{g}}\gets \text{mean-est}(\bar{\mv{g}},\mv{g},\eta_k)$;\quad $\tilde{\mv{g}}\gets \operatorname{clip}(\mv{g}-\bar{\mv{g}},C_{\mv{g}})$ \hfill (\ref{eq:mean})
     \For{$i\in B$}
     \label{line:b-block-1}
        \State \label{line:b-block-2}
        $\phi^B_i\gets \phi^B_i+\eta_k \Big(1-\frac{\tilde g_{B,i}^2}{M_i \left(\mv{\theta}_A^{(k)};\phi_i \right)}\Big)\nabla_{\phi^B_i}\log M_i \left(\mv{\theta}_A^{(k)};\phi^B_i \right)$ \hfill (\ref{eq:conditional_fisher}, \ref{eq:sgd_update_general})
     \EndFor
     \For{$i\in A$} \label{line:a-block-1}
        \State \label{line:a-block-2} $\phi^A_i\gets \phi^A_i+\eta_k \Big(1-\frac{\tilde g_{A,i}^2}{M_i \left(\phi^A_i \right)}\Big)\nabla_{\phi^A_i}\log M_i \left(\phi^A_i \right)$
     \EndFor
\EndFor
\State \textbf{Output:} $\{\mv{\theta}^{(k)}\}_{k=1}^N$.
\end{algorithmic}
\end{algorithm}
Gradient clipping is common in machine learning \cite[e.g.][]{pascanu2013difficulty}, and means that the norm of gradients are capped to a threshold $C_{\mv{g}}$, which we set to a pre-defined quantile which is itself learned adaptively; see Appendix \ref{sec:clipping-details} for details of the gradient clipping. We use the step size sequence $\eta_k = (k + n_0)^{-\kappa}$ with $n_0=5$ and $\kappa = 0.75$ in the experiments. The integrator step size $\epsilon$ used in line \ref{item:advance-rmhmc} of Algorithm \ref{alg:method_overview} is also adapted at each iteration based on the observed Hamiltonians by a Robbins--Monro stochastic approximation (with adaptive step sizes); see Appendix \ref{app:stepsize-adaptation} for details.

Mean estimation is, to our knowledge, a new method for stabilization, which can be useful in other applications, when learning score gradients in a Markov chain context. We know that \emph{under stationarity} the expected value of the gradient is exactly zero, so it may seem counter-intuitive to estimate the gradient. However, during the initial phase of the algorithm, the mean of the gradients can be atypically large, and point to the same direction for a long time. This can cause initial instability to the estimation of the mass matrix.
To mitigate such a scenario, we employ mean adaptation, that is, estimate the running mean $\mv{g}$ using standard Robbins--Monro convex combinations:
\begin{align}
\label{eq:mean}
\text{mean-est}(\bar{\mv{g}}, \mv{g}, \eta_k)
= (1-\eta_k)\bar{\mv{g}}+\eta_k \mv{g}
\end{align}
and use the `centered' gradient $\tilde{\mv{g}}$ to estimate information matrices.

We stress that Algorithm \ref{alg:method_overview} does not include parameters that require task-specific tuning of parameters, other than specifying the hierarchical mass matrix model.
In the experiments, we have tested the algorithm without some of the stabilization mechanisms; they can be turned off by using $C_{\mv{g}}=\infty$ (no clipping) and setting $\text{mean-est}(\uarg) \equiv \mv{0}$ (no mean estimation).

\section{Experiments}
\label{sec:experiment}

This section is devoted to exploring several experiments with four models, where we investigate the practical performance and properties of Algorithm \ref{alg:method_overview}. 
We first consider the Neal's funnel \eqref{eq:funnel_distribution} In Section \ref{sec:funnel-experiment}, where the optimal mass matrix is known, and verify that our estimate converges toward it; we also (unsurprisingly) find that a block structure outperforms diagonal adaptation. 
We then move to a more complex and practically relevant horseshoe-prior model, chosen because it induces dimension-wise funnels, and again find strong performance, especially under a more advanced parameterization, which account for the model has a prior and likelihood part both effecting the mass matrix; see Section \ref{subsec:mass_matrix_hier_choice} for details. 

Our third example is a stochastic volatility model with real financial dataset.
In this example, we investigate the impact of the dependence among latent variables; 
since \citep{hird-livingstone} notes that diagonal preconditioning can be worse than none, 
we examine this setting and find that both diagonal and block mass matrices improve over no preconditioning,
 with the block mass matrix again performing the best. 
 In the same example we also test whether a more general NUTS stopping criterion helps, 
 motivated by the fact that the standard criterion assumes Euclidean geometry and a generalization has been proposed; 
 here, however, the Euclidean criterion performs better, for reasons that are unclear.

Finally, we study a random-effects model with negative binomial measurements, a small extension of the example in \cite{livingstone2022barker}, and observe a clear benefit from estimating the mean when learning the mass matrix (see Section \ref{sec:full-algorithm}). 
Namely, without mean estimation, the method may fail to converge fast because the large initial gradients can severely distort the mass-matrix estimate, whereas mean estimation improves robustness.

All the experiments are implemented in Python, using the JAX library \citep{jax2018github}. Throughout the experiments, we report effective sample size (ESS) estimates computed with ArviZ library \cite{Martin2026}, which uses the Geyer's initial sequence estimator for the asymptotic variance of reversible Markov chains \cite{Geyer1992_practicalmcmc}.

\subsection{Neal's Funnel}
\label{sec:funnel-experiment}

Our first experiment is about the funnel distribution \eqref{eq:funnel_distribution}
letting the dimension of $\mv{x}$ be $20$.
This target is an extreme but representative example of hierarchical models in which the scale of one set of parameters depends strongly on another, producing a funnel-shaped geometry that is challenging for generic MCMC algorithms. 
Figure \ref{fig:funnel_example} illustrates this geometry via the joint scatter plot of $v$ and $x_1$.
For the proposed method, this model is particularly well-suited because the hierarchical structure is explicit and $v$ directly controls the scaling of $\mv{x}$; see the discussion in Section \ref{sec:centered-and-noncentered}.

We compare three algorithms:
{\bf Block exponential}: the proposed block-adaptive NUTS with adaptive step size;
{\bf Diagonal}: NUTS with adaptive diagonal mass matrix, $\mathrm{diag}(\hat{\mathcal{I}}_3)$, and adaptive step size;
and {\bf HMC}: Standard HMC with a fixed step size, chosen sufficiently small to allow the sampler to explore most of the parameter space ($\epsilon=0.01$).

For each of the methods, we run $50,000$ iterations after an initial $10,000$ burn-in steps.
Figure \ref{fig:funnel_histogram} shows that only the proposed method managed to explore the full marginal distribution of $v$, while
the HMC has under-explored the tail of $v$, and the sample of the diagonal method is not representative at all.

Table \ref{tab:ess_per_grad_comparison} reports the ESS per 1000 gradient evaluations; 
for both $v$ and $\mv{x}$, the block-adaptive method is clearly the most efficient.
 The diagonal method appears unreliable in this setting, likely because it fails to traverse the full funnel geometry.
  Finally, Figure \ref{fig:funnel_parameter} shows that the mass parameters, $\phi_i$, used to define the block \eqref{eq:exponential_model} for the Block exponential method converge towards their theoretical values derived in Example \ref{ex:funnel_model}. 
\begin{figure}
  \centering
  \includegraphics[scale=0.5]{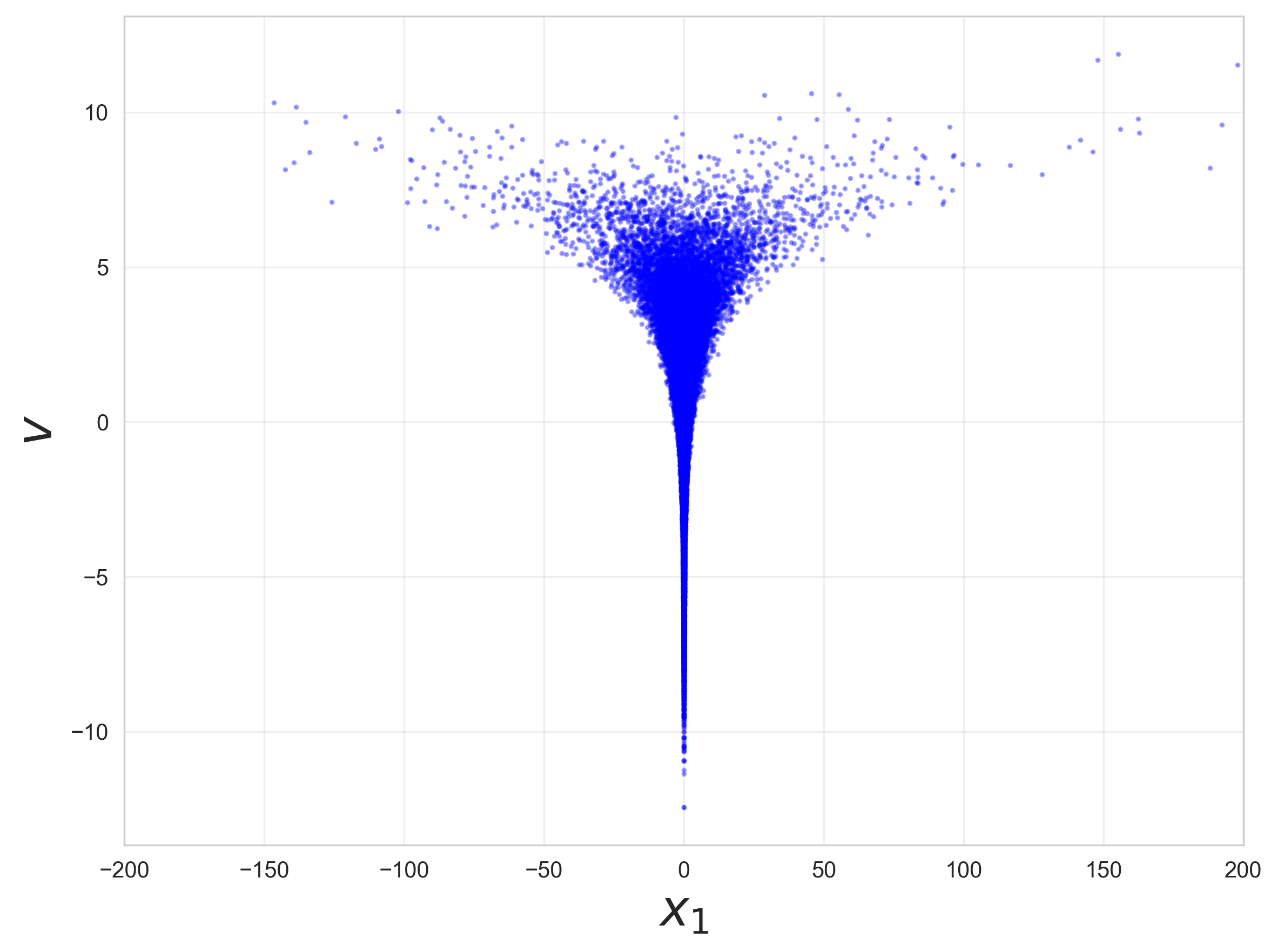}
  \caption{Samples from the funnel distribution. The y-axis is location of samples of $v$, and the x-axis is samples of $x_1$.}
  \label{fig:funnel_example}
\end{figure}

\begin{figure}
  \centering
  \includegraphics[scale=0.5]{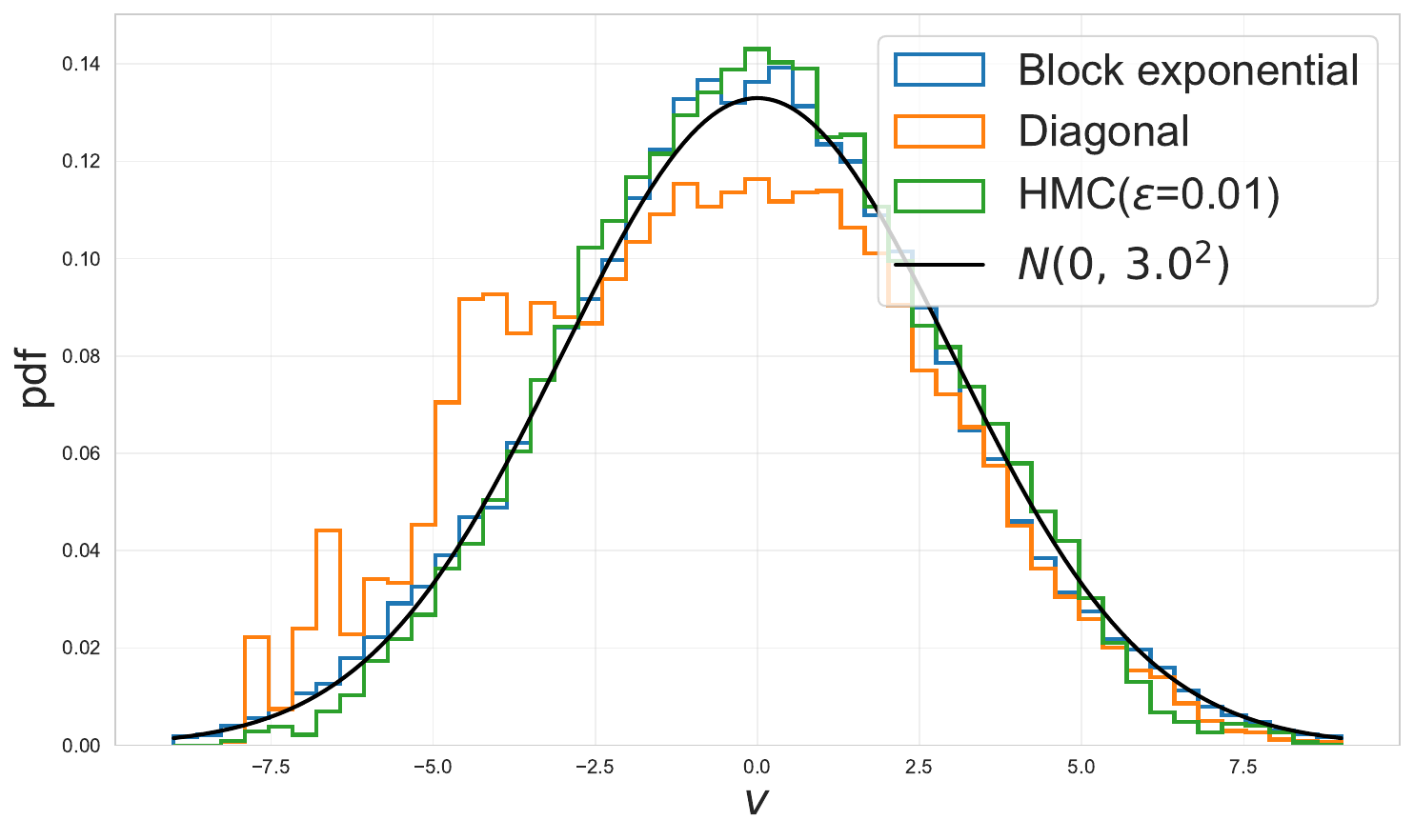}
  \caption{Histograms over the samples for $v$. Here we can see that only Block exponential algorithm actually explores the full distribution correctly.
  The $HMC$ method seems not yet explored fully the tail of distribution.}
  \label{fig:funnel_histogram}
\end{figure}

\begin{figure}
  \centering
  \includegraphics[width=0.7\textwidth]{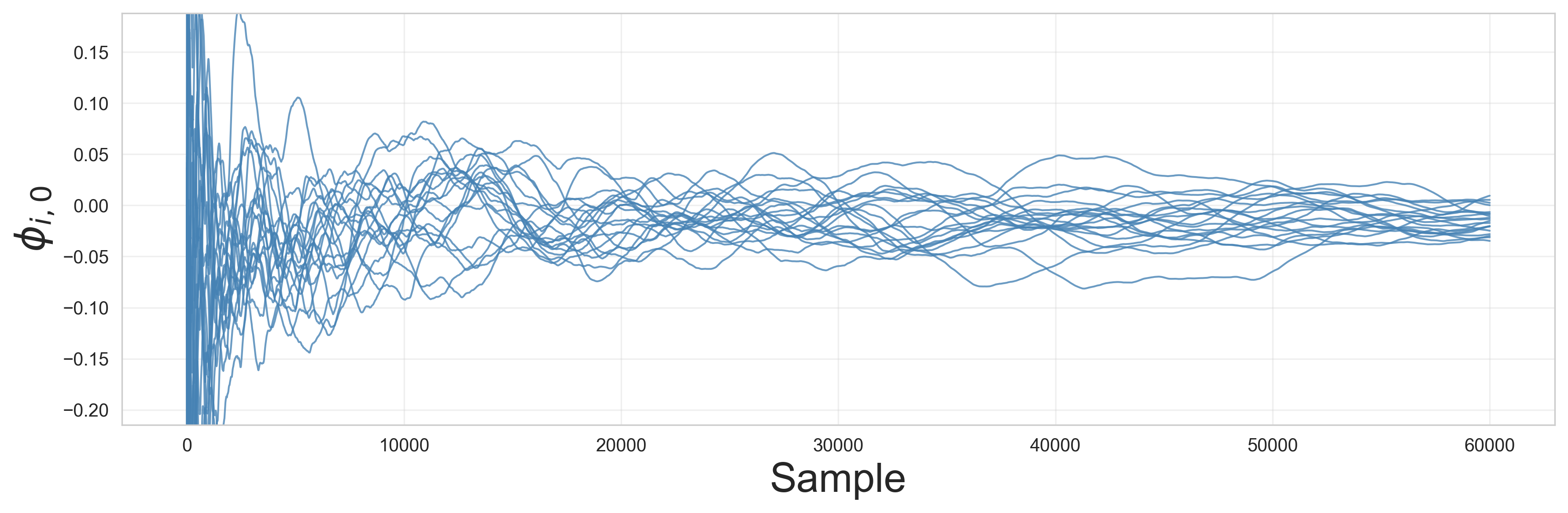}\\
  \includegraphics[width=0.7\textwidth]{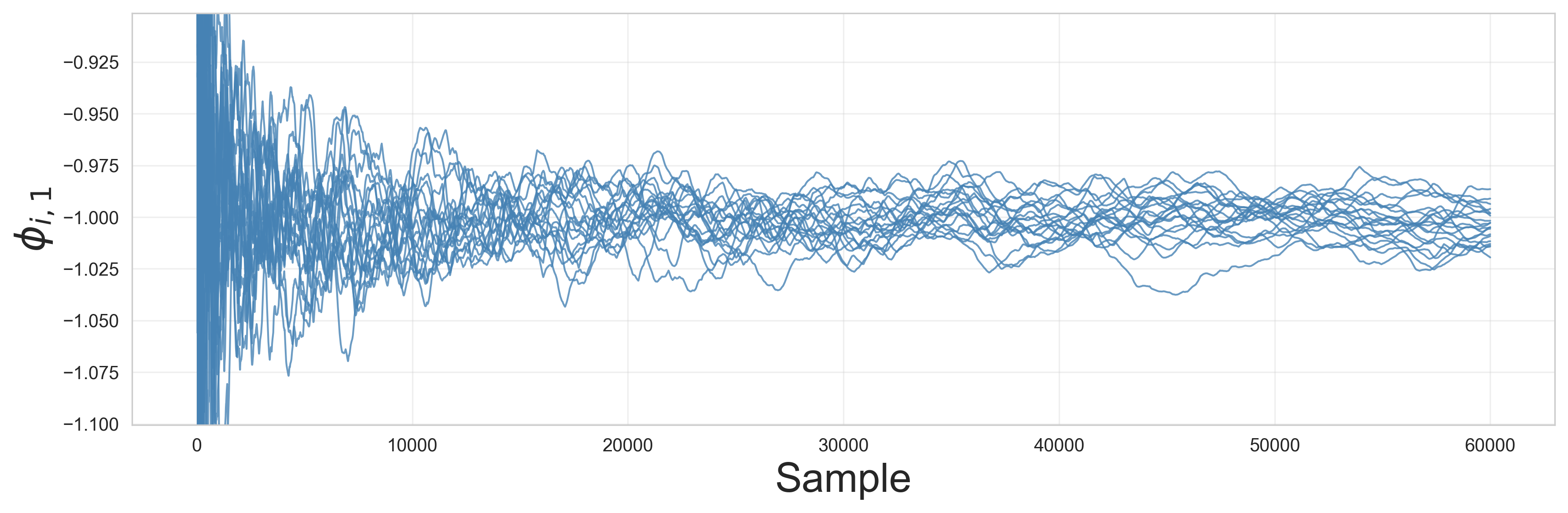}
  \caption{The trajectories of the mass parameters $\phi_{i,0}$ and $\phi_{i,1}$.}
  \label{fig:funnel_parameter}
\end{figure}

\begin{table}[h]
\centering
\caption{Effective Sample Size per 1000 gradient evaluations for different NUTS variants}
\label{tab:ess_per_grad_comparison}
\begin{tabular}{lccc}
\toprule
Method & $\# \nabla$ & $1000\frac{ESS}{\nabla}(v)$ & $\min_i 1000\frac{ESS}{\nabla}(x_i)$ \\
\midrule
Block exponential & 406,721 &  2.89 & 257 \\
Diagonal & 7,794,595 & 0.03 & 4 \\
HMC & 31,099,457 & 0.01 & 0.10 \\
\bottomrule
\end{tabular}
\end{table}

\subsection{Horseshoe Prior}
\begin{figure}
  \centering
  \includegraphics[width=0.49\linewidth]{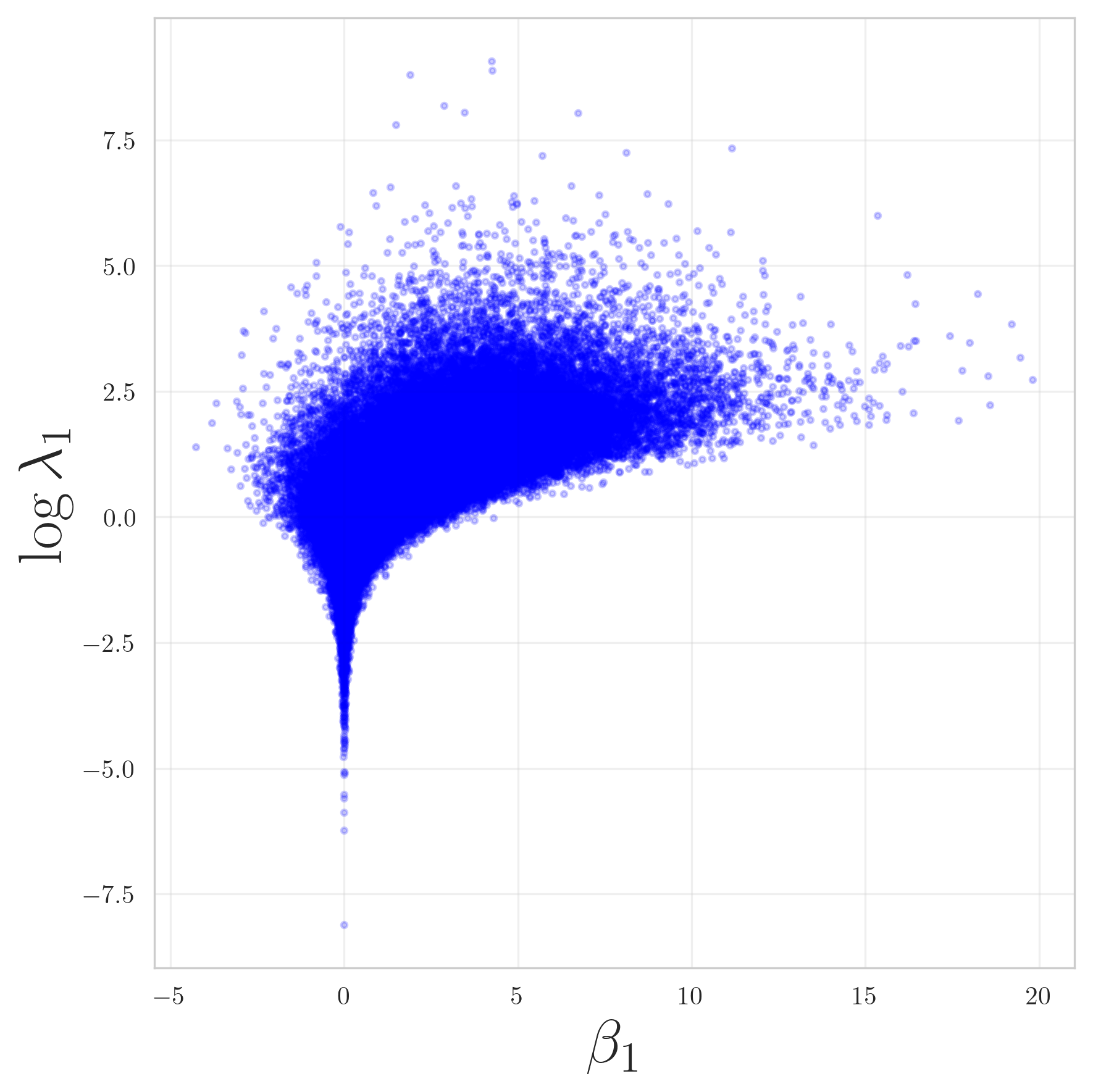}
  \includegraphics[width=0.49\linewidth]{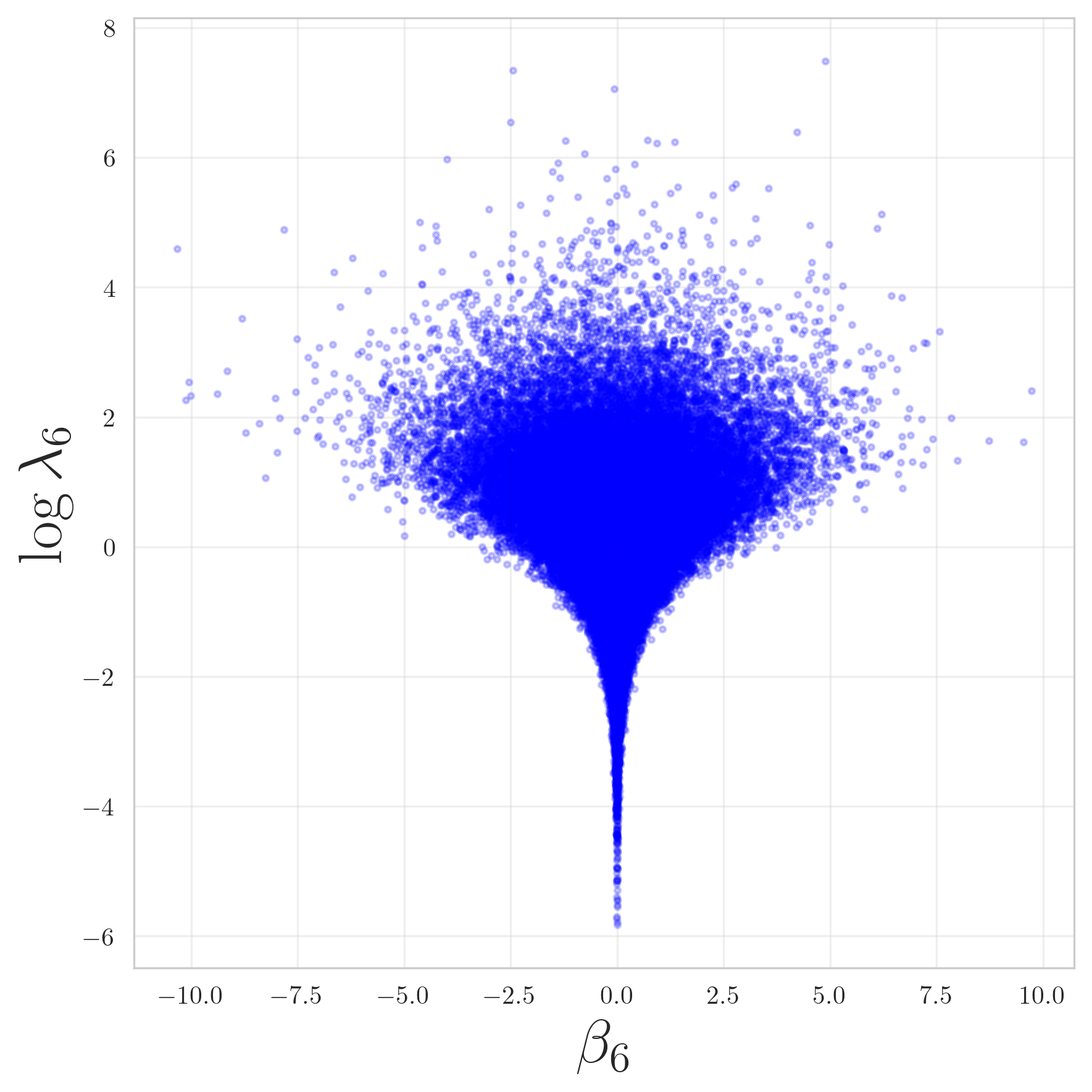}
  \caption{Samples from the horseshoe prior with additive $\beta$'s. 
  The left figure is for $\beta_1$ (true non-zero) and the right figure is for $\beta_6$ (true zero).
  We can see the influence of the likelihood influenced funnel for $\beta_1$ but a regular funnel for $\beta_6$.}
  \label{fig:horseshoe_scatter}
\end{figure}
Our second example is about Bayesian inference with the horseshoe prior, which is a common choice for sparse signal recovery in high-dimensional regression \citep{carvalho2009handling}. In a hierarchical parameterization, each regression coefficient is paired with a local scale parameter, which creates a collection of local funnel geometries. In the logistic regression setting, which we consider, the data $y_1,\ldots,y_n\in\{0,1\}$ are modeled as follows:
\begin{align*}
y_i \mid \eta_i &\sim \text{Bernoulli}(\text{logit}^{-1}(\eta_i)), \\
\eta_i &= \beta_0 + \mathbf{x}_i^T \boldsymbol{\beta}, \qquad i=1,\ldots,n,
\end{align*}
where $x_i\in\mathbb{R}^p$ are covariates, and $\beta_0\in\mathbb{R}$ and $\boldsymbol{\beta}\in\mathbb{R}^p$ are the parameters, which have the following horseshoe prior:
\begin{align*}
\beta_j \mid \lambda_j &\sim \mathcal{N}(0, \tau^2 \lambda_j^2), \\
\lambda_j &\sim \text{Half-Cauchy}(0,1), \qquad j=1,\ldots,p, \\
\beta_0 &\sim \mathcal{N}(0,\sigma_0^2),
\end{align*}
with hyper-parameters $\tau,\sigma_0^2>0$.

Under this prior, the conditional distribution of $\beta_j$ forms a funnel with respect to $\lambda_j$.
The likelihood can strongly reshape this geometry for signals with non-negligible contribution, as illustrated in Figure \ref{fig:horseshoe_scatter}, where the joint distribution of $(\beta_1,\lambda_1)$ is skewed due to the likelihood influence, while $(\beta_6,\lambda_6)$ remains close to the prior shape. This suggests that the sum-of-exponentials form in \eqref{eq:sum_of_exponentials} could be better suited than the single-exponential form in \eqref{eq:exponential_model} for modeling the metric associated with the $\beta$-coordinates. 

For this experiment we set $n=100$ and $p=20$. We generate synthetic predictors $\mathbf{x}_i \sim \mathcal{N}(0,\Sigma)$ with $\Sigma_{ii}=1$ and $\Sigma_{ij}=0.8$ for $i\neq j$. The true coefficient vector has the first five entries equal to one and the remaining entries equal to zero, and we generate data $y_i \sim \text{Bernoulli}(\text{logit}^{-1}(\eta_i))$ with $\eta_i = \mathbf{x}_i^T \boldsymbol{\beta}$.

We compare three methods: {\bf Block sum of exponentials}, the proposed block-adaptive NUTS with adaptive step size using \eqref{eq:sum_of_exponentials}; {\bf Block exponential}, the same method using the single-exponential form \eqref{eq:exponential_model}; and {\bf Diagonal}, NUTS with an adaptive diagonal mass matrix and adaptive step size. In the block methods we use a fixed block structure that groups $\log(\lambda_j)$ and $\beta_0$, while the metric for $\boldsymbol{\beta}$ is parameterized either by the exponential or sum-of-exponentials form:
\begin{align*}
  M^1_i(\boldsymbol{\lambda}) &= \exp\!\left(\phi_{i,1}+ \log(\lambda_i)\phi_{i,2}\right), \\
  M^2_i(\boldsymbol{\lambda}) &= \exp\!\left(\phi^1_{i,1}+ \log(\lambda_i)\phi^1_{i,2}\right) + \exp\!\left(\phi^2_i\right).
\end{align*}
We run each algorithm with adaptive step size for $50,000$ iterations after an initial $10,000$ burn-in steps.
The effect of the metric choice is evident in the trace plots of $\log(\lambda_1)$ in Figure \ref{fig:traceplot_log_lambda_1}. 
The diagonal metric shows a high proportion of divergent transitions (8.9\%),
 with a transition classified as divergent when the kinetic energy changes by more than 1000 units in the trajectory.
  The single-exponential block form still yields a substantial proportion of divergent transitions (2.3\%). 
  By contrast, the sum-of-exponentials form produces only 0.01\% divergent transitions, with these occurring almost exclusively during the initial burn-in and adaptation phase.

This difference can be explained by the fitted functional forms. Empirically, the exponential form typically behaves like
\[
M^1_i(\boldsymbol{\lambda}) \approx \exp\!\left(0.8 - 0.8\,\log(\lambda_i)\right),
\]
whereas the sum-of-exponentials form behaves like
\[
M^2_i(\boldsymbol{\lambda}) \approx \exp\!\left(-1.8\,\log(\lambda_i)\right) + \exp(-0.5),
\]
with considerable variability in the fitted parameters. The trajectories of the parameters are shown in Appendix \ref{app:horseshoe-trace}.
 The additional constant term in $M^2_i$ prevents the metric from becoming excessively large when $\log(\lambda_i)$ is very negative (i.e., when $\lambda_i$ is extremely small), which in turn reduces numerical instability and suppresses divergent transitions.

Finally, Table \ref{tab:horseshoe_ess_per_grad} reports the effective sample size for representative parameters over the 50{,}000 samples. 
The diagonal method is the least efficient, the block-exponential approach is substantially more efficient, and the sum-of-exponentials parameterization provides slightly better overall performance than the block-exponential.

\begin{table}[h]
\centering
\caption{Effective sample size (ESS) per 1000 gradient evaluations for representative horseshoe parameters under three adaptive NUTS metrics. Larger values indicate higher efficiency.}
\label{tab:horseshoe_ess_per_grad}
\begin{tabular}{l r c c c}
\toprule
Method & $\# \nabla$ & $\frac{1000\,ESS}{\nabla}(\beta_0)$ & $\min_i \frac{1000\,ESS}{\nabla}(\beta_i)$ & $\min_i \frac{1000\,ESS}{\nabla}(\log \lambda_i)$ \\
\midrule
Block sum of exponential & 2{,}290{,}530  & 17.31 & 4.09 & 6.56 \\
Block exponential        & 2{,}471{,}758  & 12.07          & 2.97          & 4.59         \\
Diagonal                 & 13{,}328{,}706 & 0.23           & 0.11          & 0.09          \\
\bottomrule
\end{tabular}

\vspace{0.25em}
\end{table}

\begin{figure}
  \centering
   \includegraphics[width=1\linewidth]{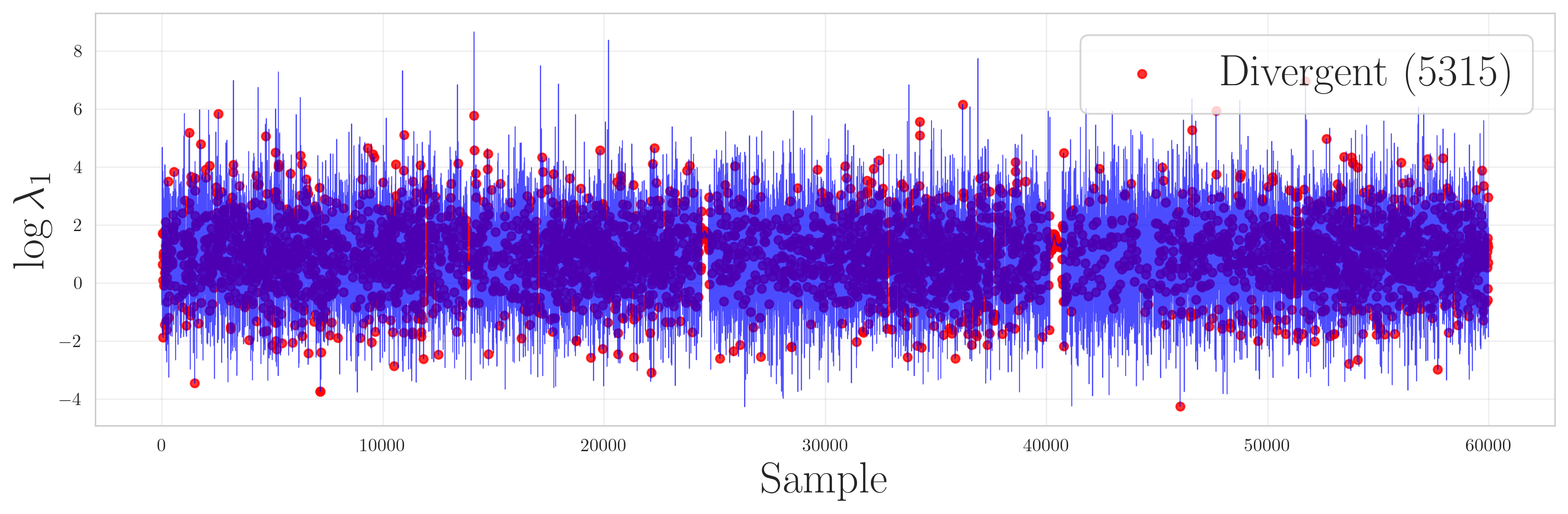}
  \includegraphics[width=1\linewidth]{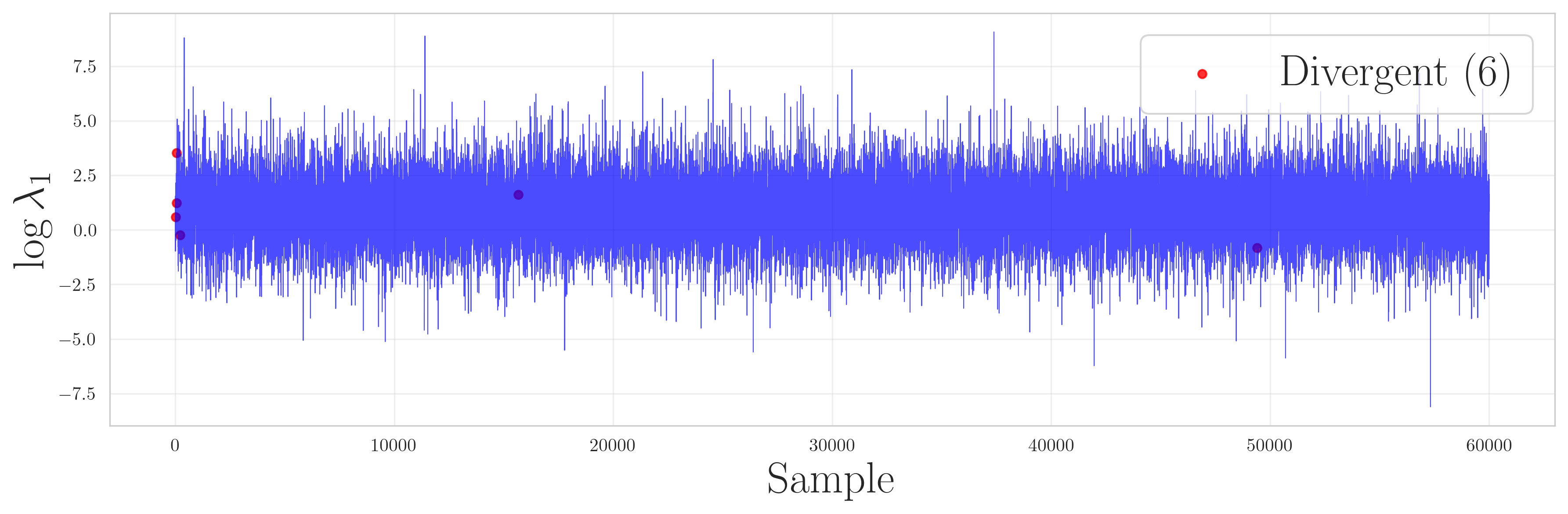}
  \includegraphics[width=1\linewidth]{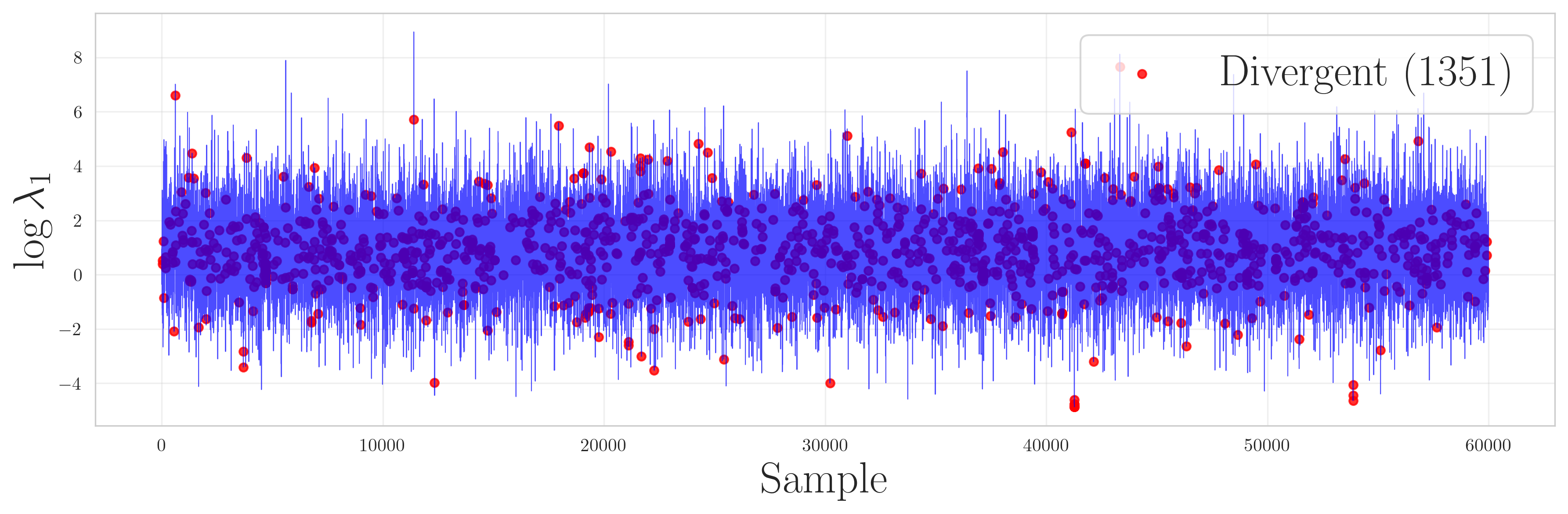}
  \label{fig:traceplot_log_lambda_1}
  \caption{Trace plot for $\log(\lambda_1)$ for the three methods: Diagonal (top),   Block exponential (middle), and Block sum-of-exponentials (bottom).}
\end{figure}

\subsection{Stochastic volatility model}
In this example, we consider the stochastic volatility (SV) model, which is a common nonlinear state-space model used in financial applications \citep[e.g.][]{shephard-selected}, and also often used to benchmark MCMC algorithms.
The data we use are 480 observations of centered monthly log returns for the S\&P 500 between 1978 and 2017.
The SV model we consider has a latent log-volatility process $x_t$ that follows a priori stationary autoregressive dynamics. The values of $x_t$ scale the instantaneous variance of the observed data process $y_t$. The full model is given by:
\begin{align*}
\kappa 
&\sim \log \mathcal{N}(-2, 1),
\\
\phi^\star 
&\sim \mathcal{N}(0, 2),
&
\phi &= \frac{e^{\phi^\star} - 1}{e^{\phi^\star} + 1} \in (-1,1),
\\
\sigma^2 
&\sim \mathrm{InvGamma}(4, 4),
\\
x_0& \mid \phi, \sigma^2 
\sim \mathcal{N}\!\left(0, \frac{\sigma^2}{1 - \phi^2}\right),
\\
x_t& \mid x_{t-1}, \phi, \sigma^2 
\sim \mathcal{N}(\phi x_{t-1}, \sigma^2),
&
t &= 1,\dots,T,
\\
y_t& \mid x_t, \kappa 
\sim \mathcal{N}\!\bigl(0, \kappa^2 e^{x_t}\bigr),
&
t &= 1,\dots,T.
\end{align*}
Here we use the three block structure in a sequential order: block one is $\phi$, block two is $(\kappa,\sigma^2)$, and block three is $\mv{x}=(x_0,\dots,x_T)$.
For block two we use a single-exponential model of the form \eqref{eq:exponential_model}, and for block three we use a sum-of-exponentials model of the form \eqref{eq:sum_of_exponentials}.
We consider this specific model to investigate two questions.
First, \cite{hird-livingstone} reports that using $\mathrm{diag}(\mathrm{Cov}(\pi))^{-1}$ as a mass matrix 
can degrade performance 
relative to using the identity mass. Since we employ a diagonal mass matrix, we want to assess whether a similar effect occurs when the posterior exhibits strong parameter dependence. 
We emphasize that $\mathrm{diag}(\mathrm{Cov}(\pi))^{-1}$ and $\mathrm{diag}(\mathcal{I}_3)$ 
are fundamentally different choices, and it is unclear whether the latter inherits the same theoretical pathology as the former.

Our second objective is to compare different NUTS stopping criteria: as noted in the appendix of \cite{Betancourt2017}, 
the standard NUTS termination rule is not the most natural under non-Euclidean metrics. 
We therefore implement the generalized stopping criterion proposed in \cite{Betancourt2013_rmnuts} (see Appendix \ref{sec:generalized_nouturn} for details) and compare it to the usual criterion. To isolate the effect of the metric and the stopping rule, 
  we consider four methods: Generalized block (exponential block adaptation with the generalized stopping criterion),
   Euclidean block (exponential block adaptation with the standard stopping criterion), 
   Diagonal NUTS (diagonal metric adaptation), and regular NUTS (no metric adaptation).

  We run each algorithm with adaptive step size for $50,000$ iterations after an initial $10,000$ burn-in steps.
 The effective sample size per 1000 gradient evaluations is reported in Table \ref{tab:ess_per_grad_stochvol}. 
  Somewhat surprisingly, the generalized stopping criterion does not improve performance, but rather the Euclidean block variant performs best with this specific model. We do not have a clear explanation why this happens. Nevertheless, both block methods substantially outperform the diagonal and regular variants,
   particularly for $\phi$ and $\sigma^2$, indicating that capturing posterior dependence at the block level can yield 
   large efficiency gains.

\begin{table}[h]
\centering
\caption{Effective sample size (ESS) per gradient evaluations for NUTS variants on the stochastic volatility model. Larger values indicate higher efficiency.}
\label{tab:ess_per_grad_stochvol}
\begin{tabular}{l r c c c c}
\toprule
Method & $\# \nabla$ & $\frac{1000ESS}{\nabla}(\kappa)$ & $\frac{1000ESS}{\nabla}(\phi)$ & $\min_i \frac{1000ESS}{\nabla}(\sigma^2) $ & $\min_i \frac{1000ESS}{\nabla}(x_i) $ \\
\midrule
Euclidean block   & 2{,}615{,}210 & 2.15 & 1.48 & 1.63 & 5.80 \\
Generalized block & 3{,}805{,}018 & 1.75          & 0.86          & 0.86          & 2.50          \\
Diagonal          & 3{,}034{,}262 & 1.43          & 0.36          & 0.29          & 1.85          \\
Regular NUTS      & 5{,}468{,}470 & 1.08          & 0.27          & 0.20          & 1.71          \\
\bottomrule
\end{tabular}

\vspace{0.25em}
\end{table}

\subsection{Negative Binomial}
\label{sec:ex_nb}
Our final example is based on the hierarchical Poisson model of \cite{livingstone2022barker}, except that we replace the Poisson distribution with a negative binomial distribution to allow for overdispersion,
 making the resulting inference problem a bit more challenging. The model is\begin{align*}
  \nu &\sim \operatorname{InvGamma}(1, 0.5),\\
  \mu &\sim \mathcal{N}(0, 10^2),\\
  \eta_i \mid \mu &\sim \mathcal{N}(\mu, 3^2), \qquad i=1,\dots,50,\\
  y_{ij}\mid \eta_i,\nu 
  &\sim \mathrm{NegBin}\!\left(\nu,\; \frac{\nu}{\nu+\exp(\eta_i)}\right),
  \qquad i=1,\dots,50,\; j=1,\dots,5.
\end{align*}
We use the size-probability parameterization of the negative binomial, so that
\(
\mathbb{E}[y_{ij}\mid \eta_i,\nu]=\exp(\eta_i)
\)
and
\(
\operatorname{Var}(y_{ij}\mid \eta_i,\nu)=\exp(\eta_i)+\exp(2\eta_i)/\nu.
\)
Data are generated with true parameters \(\mu=10\) and \(\nu=0.5\).

This configuration produces large mean counts and, when the chain is not initialized near its stationary distribution, leads to large gradients.
Therefore, it provides a challenging test case for the use of initial gradients in parameter mass estimation.
We examine our strategy of centering the gradient observations using an estimated mean \eqref{eq:mean}, and compare it against the baseline without mean correction.

For block structure we have the hyperparameters for \(\mu\) and \(\nu\) as block one, and \(\eta_i\)'s as block two.
As before, we compare three strategies (for the second block): {\bf Block sum of exponentials}, {\bf Block exponential}, and {\bf Diagonal}. Our main focus is on the difference between {\bf Block exponential} with mean correction and the variant that enforces \(\bar g = 0\). We run \(10{,}000\) sampling iterations, as our primary interest is in the burn-in phase.
Figure~\ref{fig:traceplot_mu_negbin} displays trace plots for \(\mu\). 
 The variant without mean adaptation fails to reach the stationary regime within this period. While all the methods using the mean adaptation reaches stationarity.
 We can also see that diagonal method reaches the stationary state faster than the block exponential which in turn is faster than the block sum of exponentials.
 Indicating the increasing level of complexity of the mass matrix slows down the initial convergence. 
 Table~\ref{tab:negbin_ess_per_grad_comparison} reports the effective sample size per \(1000\) 
 gradient evaluations for the working methods; {\bf Block exponential} attains the highest efficiency across all reported parameters.
 The difference is much smaller than in the examples with funnel but still significant.

\begin{figure}
  \centering
  \includegraphics[width=0.49\linewidth]{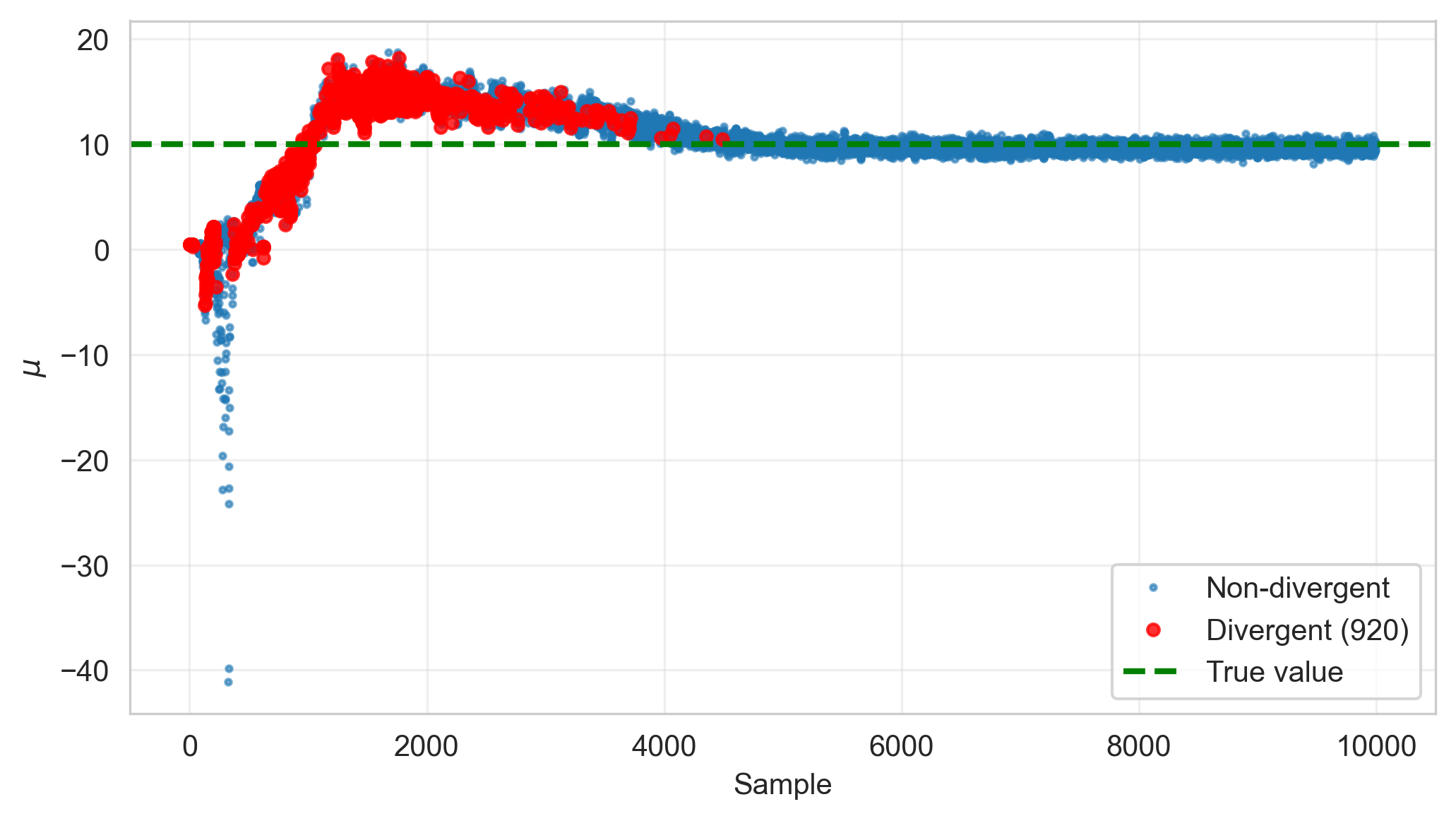}
  \includegraphics[width=0.49\linewidth]{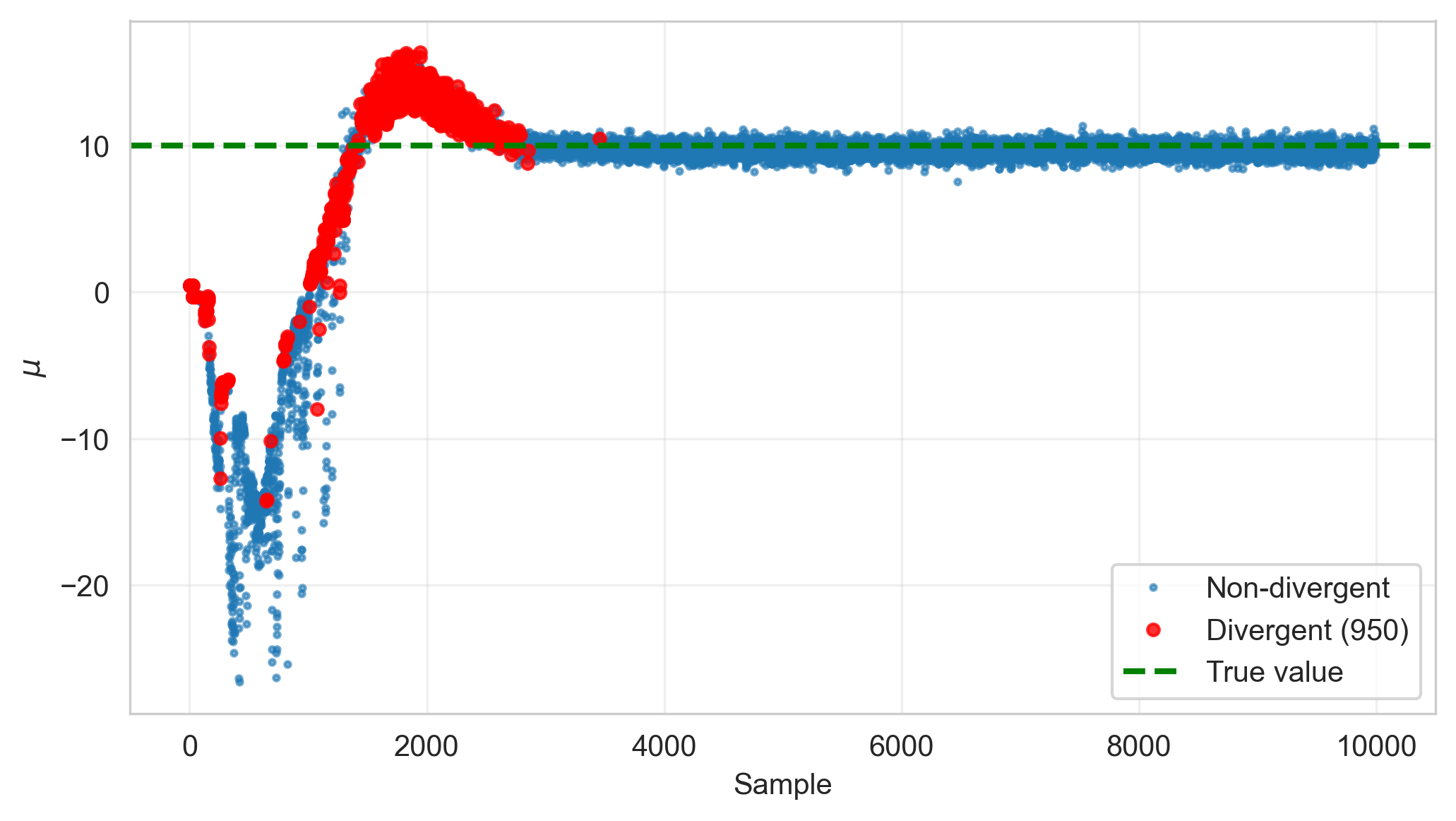}
  \includegraphics[width=0.49\linewidth]{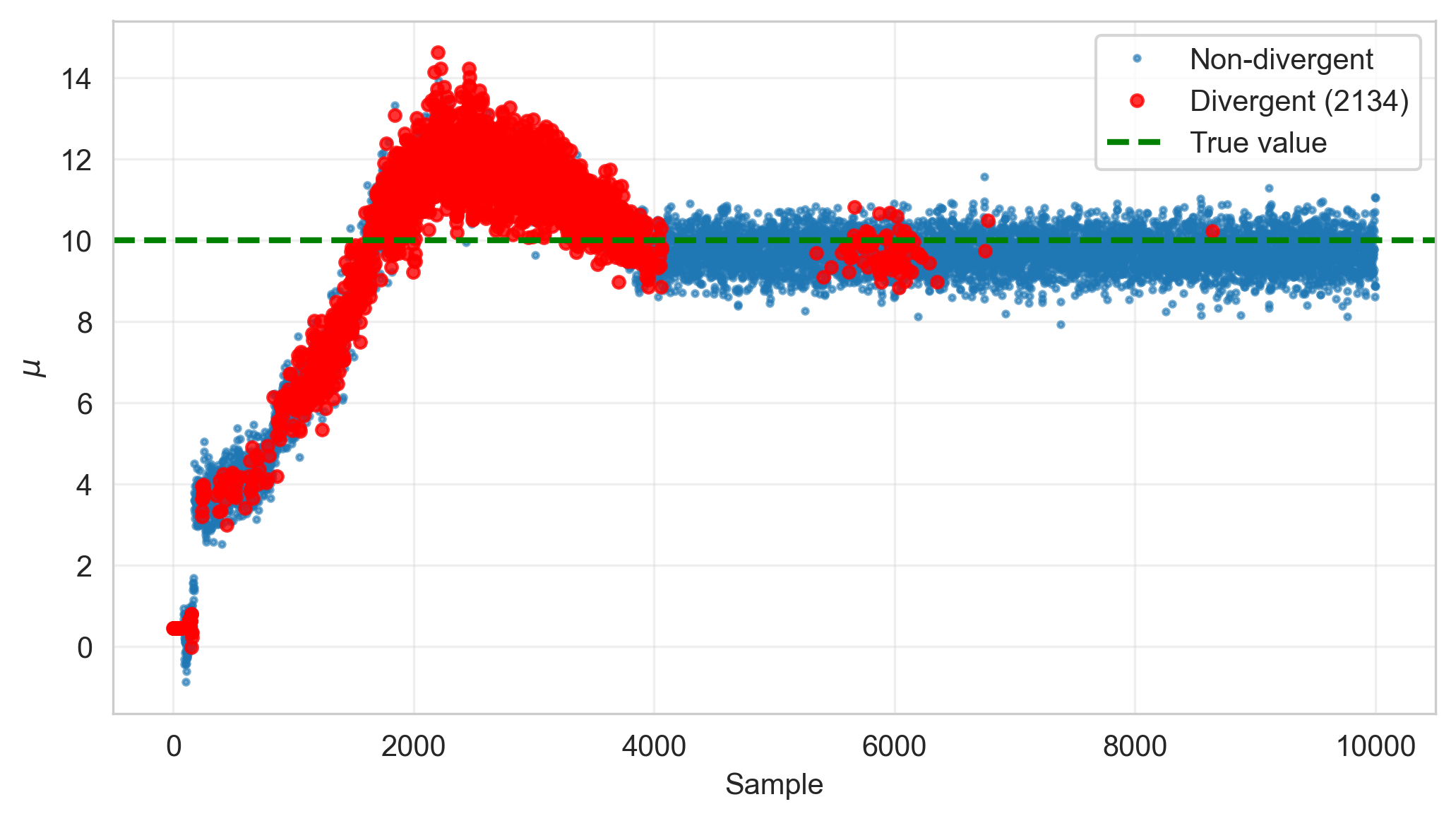}
  \includegraphics[width=0.49\linewidth]{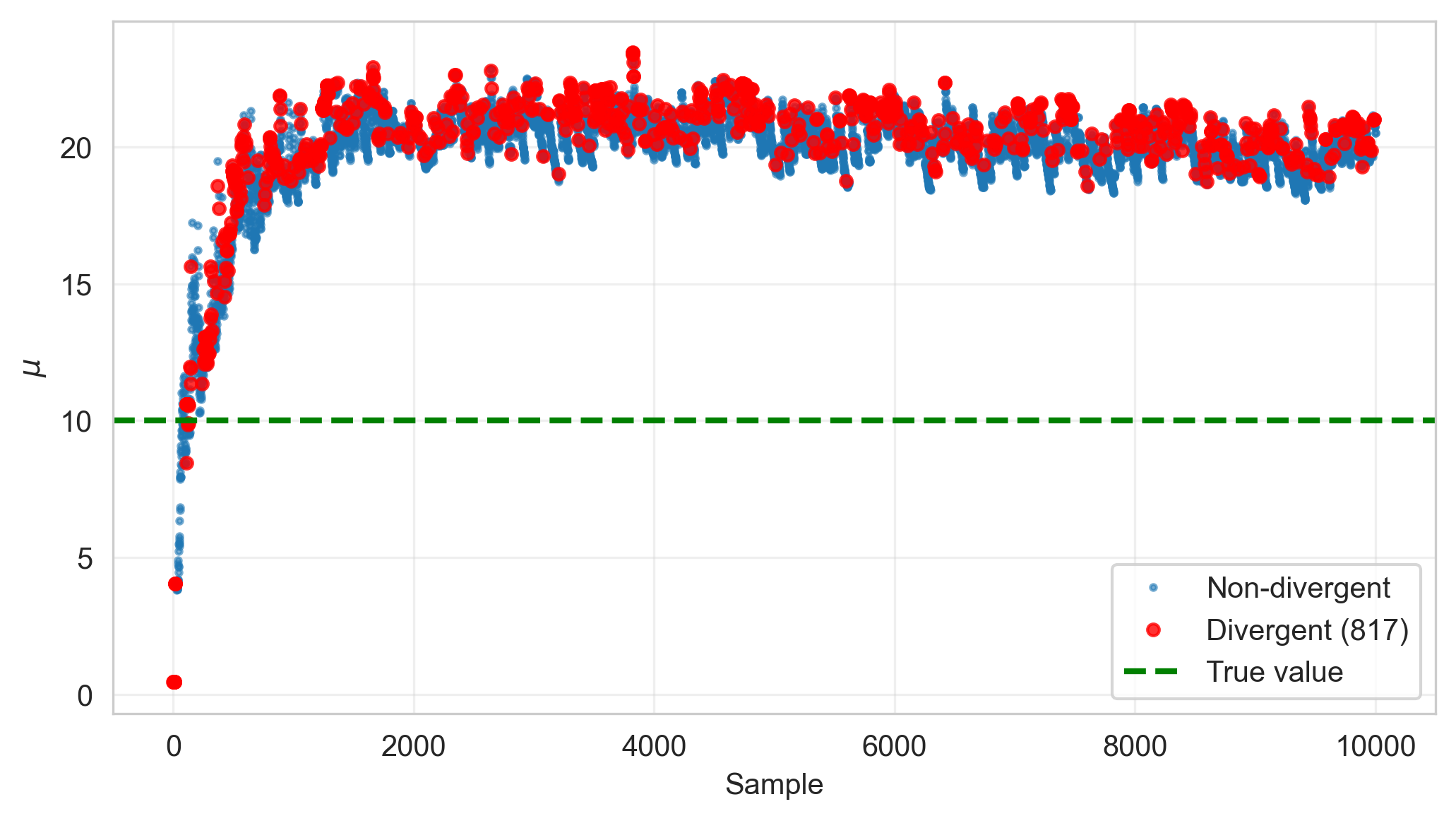}
  \caption{Trace plots for \(\mu\): Block sum of exponentials (top left), Diagonal (bottom left), Block exponential (top right), and Block exponential  with no mean adaptation (bottom right).}
  \label{fig:traceplot_mu_negbin}
\end{figure}

\begin{table}[h]
\centering
\caption{Effective sample size per gradient evaluations for the negative binomial model, computed from the last \(40{,}000\) samples.}
\label{tab:negbin_ess_per_grad_comparison}
\begin{tabular}{lcccc}
\toprule
Method & $\# \nabla$ & $\frac{1000ESS}{\nabla}(\mu)$ & $\frac{1000ESS}{\nabla}(\log \nu)$ & $\min_i \frac{1000ESS}{\nabla}(\eta_{i})$  \\
\midrule
Block sum of exponentials & 826,975 & 27.05 & 25.82 & 53.49  \\
Block exponential & 1,273,347 & 15.27 & 11.70 & 33.26  \\
Diagonal & 3,142,553 & 15.77 & 9.13 & 6.46  \\
\bottomrule
\end{tabular}
\end{table}

\section{Discussion}
\label{sec:discussion}

We present a novel hierarchical RMHMC method which admits an explicit symplectic splitting integrator, and an adaptive algorithm which learns its parameters during estimation. The method is well-suited for many hierarchical models, exploiting their structure to devise a learnable position-dependent mass matrix. 
In order to keep the approach computationally efficient and easy to implement, we restricted our attention to block-diagonal (and, in our experiments, diagonal) mass matrices, which allow for particularly efficient algorithm.
We stress that the target density need not possess a block structure, or indeed any hierarchical structure at all. Rather, the hierarchical decomposition is a modeling choice imposed on the mass matrix, viewed as a surrogate for the local geometry of the target distribution.

The framework nonetheless admits natural extensions to richer matrix families. A simple generalization is to introduce a linear preconditioning transform of the form
\[
\mv{B}\,\mv{M}_{\mv{\phi}}(\mv{\theta}_A)\,\mv{B}^\top,
\]
where \(\mv{M}_{\mv{\phi}}(\mv{\theta}_A)\) is diagonal (and depends on the hierarchical component \(\mv{\theta}_A\)). The matrix \(\mv{B}\) can either be chosen as lower triangular (to allow a Cholesky-like parameterization) or orthonormal (to act as a rotation). This retains computational tractability while allowing for non-diagonal structure in the effective metric. 
Furthermore, one can relate properties of the approximating density to properties of the approximating matrix. 
For instance, a diagonal approximation corresponds to assuming that the coordinates of the gradient $\mv{g}$ are independent.
While one may also assume a low-rank dependence structure among the elements of $\mv{g}$.
Indeed, in a concurrent work \citeauthor{monahan-kristensen-thorson-carpenter} \citep{monahan-kristensen-thorson-carpenter} investigate learning a sparse precision matrix and using it as the mass matrix.
 Similar structures could also be useful in the position-dependent setting. 
 Further, related work by \citeauthor{hird2026highdimensional}
\citep{hird2026highdimensional} using a sparse parametrizations of
dense preconditioners may be promising to explore in conjunction with our approach.

An interesting direction for future work would be to combine this idea with the general block structure described in Appendix~\ref{sec:symmetry_verification}, in which each block metric is allowed to depend on the remaining position blocks.
On the estimation side, we approximate the distribution of gradients by minimizing a Kullback--Leibler divergence between the true conditional gradient distribution and a Gaussian surrogate. 
The specific modeling choice with a Gaussian is convenient, but not necessary. For instance, replacing the normal approximation with a heavy-tailed distribution (such as a multivariate \(t\)) could yield more robust estimators of the mass matrix when gradients exhibit outliers or heavy tails; see \citep{livingstone2019kinetic} for discussion about HMC with non-Gaussian momentums.
Finally, Algorithm \ref{alg:method_overview} uses the gradients calculated at the Markov chain's state $\mv{\theta}_k$. 
It is possible to use all gradients on the integrator's path, weighting them according to their selection probabilities, similar to the `Rao--Blackwellized' adaptation of random-walk MCMC in \cite{andrieu2008tutorial}.

The experiments reported in Section \ref{sec:experiment} demonstrate that the hierarchical RMHMC method can provide substantial speedups, and the adaptive RMHMC works reliably in a range of models, without prior tuning of (hyper)parameters. Finally, the negative binomial example (Section~\ref{sec:ex_nb})
 illustrated that the mean estimation can be crucial for stable estimation. We note that since the stationary mean is known to be zero, the deviations of the estimated mean from zero may also serve as a practical diagnostic of how far the chain is from stationarity.

\section*{Acknowledgements}

M.K. and M.V. were supported by Research Council of Finland grant ‘FiRST (Finnish Centre of Excellence in Randomness and Structures)’ (346311, 364216). The authors wish to acknowledge CSC—IT Center for Science, Finland, for computational resources.

\bibliography{myrefs}
\appendix

\section{Symmetry verification of the hierarchical leapfrog integrator}
\label{sec:symmetry_verification}

\begin{algorithm}
\caption{Exact block flow $e^{h\mathcal H_k}$}
\label{alg:block_flow}
\begin{algorithmic}[1]
\Require Block index $k$, step size $h$, current state $(\mv{\theta},\mv{p})$
\State $\mv{v}_k \gets \mv{M}_k(\mv{\theta}_{-k})^{-1}\mv{p}_k$ \hfill (\ref{eq:v_flow})
\State $\mv{\theta}_k \gets \mv{\theta}_k + h\,\mv{v}_k$
\For{$j=1$ to $K$}
  \If{$j \neq k$}
    \State $\mv{p}_j \gets \mv{p}_j + \frac{h}{2}\,\mv{s}_{j\leftarrow k}(\mv{\theta}_{-k},\mv{v}_k)$ \hfill (\ref{eq:s_flow})
  \EndIf
\EndFor
\State \Return $(\mv{\theta},\mv{p})$
\end{algorithmic}
\end{algorithm}

\begin{algorithm}
\caption{Hierarchical generalized leapfrog step $\Psi_\epsilon^\pi$}
\label{alg:hier_lf_step_pi}
\begin{algorithmic}[1]
\Require Current state $(\mv{\theta}^{(n)},\mv{p}^{(n)})$, step size $\epsilon$, block order $\pi=(\pi_1,\dots,\pi_K)$
\State $(\mv{\theta},\mv{p}) \gets (\mv{\theta}^{(n)},\mv{p}^{(n)})$

\State $\mv{p} \gets \mv{p} - \frac{\epsilon}{2}\,\nabla_{\mv{\theta}} H_0(\mv{\theta})$
\For{$r=1$ to $K-1$}
  \State $(\mv{\theta},\mv{p}) \gets \Call{BlockFlow}{\pi_r,\epsilon/2,\mv{\theta},\mv{p}}$ \hfill (Algorithm~\ref{alg:block_flow})
\EndFor

\State $(\mv{\theta},\mv{p}) \gets \Call{BlockFlow}{\pi_K,\epsilon,\mv{\theta},\mv{p}}$ \hfill (Algorithm~\ref{alg:block_flow})

\For{$r=K-1$ downto $1$}
  \State $(\mv{\theta},\mv{p}) \gets \Call{BlockFlow}{\pi_r,\epsilon/2,\mv{\theta},\mv{p}}$ \hfill (Algorithm~\ref{alg:block_flow})
\EndFor

\State $\mv{p} \gets \mv{p} - \frac{\epsilon}{2}\,\nabla_{\mv{\theta}} H_0(\mv{\theta})$

\State \Return $(\mv{\theta},\mv{p})$
\end{algorithmic}
\end{algorithm}

In this section, we prove Theorem~\ref{thm:multi_symplectic_reversible_hier_leapfrog}, showing that Algorithm~\ref{alg:hier_lf_step_pi} defines a symmetric (time-reversible) and symplectic map. As Algorithm~\ref{alg:hier_lf_step} is a special case of Algorithm~\ref{alg:hier_lf_step_pi}, Theorem~\ref{thm:symplectic_reversible_hier_leapfrog} follows immediately.

The key idea is to view the explicit update as a \emph{splitting integrator}. In the present multi-block setting, we decompose the Hamiltonian as
\[
H = H_0 + \sum_{k=1}^K H_k,
\]
where each partial Hamiltonian has an exact Hamiltonian flow that can be written in closed form. As in the two-block case, it is convenient to describe these flows using operators that update either $\mv{\theta}$ or $\mv{p}$. A single integrator step is then obtained by composing the corresponding exact subflows in a palindromic order.
Assume the parameter and momentum vectors are partitioned as $\mv{\theta} = (\mv{\theta}_1,\dots,\mv{\theta}_K)$ and $\mv{p} = (\mv{p}_1,\dots,\mv{p}_K)$, and that the mass matrix is block diagonal with block $k$ depending only on the remaining position blocks,
\[
  \mv{M}(\mv{\theta})
  =
  \operatorname{blkdiag}\!\big(
    \mv{M}_1(\mv{\theta}_{-1}),\dots,\mv{M}_K(\mv{\theta}_{-K})
  \big),
\]
where $\mv{\theta}_{-k}$ denotes all blocks except $\mv{\theta}_k$. We split the Hamiltonian as
$
  H(\mv{\theta},\mv{p}) = H_0(\mv{\theta}) + \sum_{k=1}^K H_k(\mv{\theta},\mv{p}),
$
with $H_0(\mv{\theta}) = U(\mv{\theta}) + \frac12 \log\det \mv{M}(\mv{\theta})$ and
$H_k(\mv{\theta},\mv{p}) = \frac12 \mv{p}_k^\top \mv{M}_k(\mv{\theta}_{-k})^{-1}\mv{p}_k$.

Let $e^{\epsilon\mathcal{H}_\bullet}$ denote the time--$\epsilon$ Hamiltonian flow of $H_\bullet$. The exact $H_0$-subflow is the momentum kick
\[
  e^{\epsilon\mathcal{H}_0}:\; (\mv{\theta},\mv{p})
  \mapsto
  \bigl(\mv{\theta},\;\mv{p} - \epsilon\,\nabla_{\mv{\theta}}H_0(\mv{\theta})\bigr).
\]
To build $ e^{\epsilon\mathcal{H}_k}$-subflow, $k \in \{1,\dots,K\}$ is a bit more involved.
Define 
\begin{align}
  \label{eq:v_flow}
  \mv{v}_k
  &= \mv{M}_k(\mv{\theta}_{-k})^{-1}\mv{p}_k, \\
  \label{eq:s_flow}
  \big[\mv{s}_{j\leftarrow k}(\mv{\theta}_{-k},\mv{v}_k)\big]_i
  &=
  \mv{v}_k^\top
  \big(\partial_{\theta_{j,i}}\mv{M}_k(\mv{\theta}_{-k})\big)\,\mv{v}_k .
\end{align}
We then introduce the drift operator $D_k^\epsilon$, which moves only $\mv{\theta}_k$, and the kick operator $K_k^\epsilon$, which moves only the momentum blocks $\mv{p}_j$ with $j \neq k$:
\begin{align*}
  D_k^\epsilon:\; (\mv{\theta},\mv{p})
  &\mapsto
  \bigl(\mv{\theta}_1,\dots,\mv{\theta}_k + \epsilon\,\mv{v}_k,\dots,\mv{\theta}_K,\;\mv{p}\bigr), \\[2mm]
  K_k^\epsilon:\; (\mv{\theta},\mv{p})
  &\mapsto
  \bigl(\mv{\theta},\;\mv{p}_1,\dots,\mv{p}_j + \tfrac{\epsilon}{2}\,\mv{s}_{j\leftarrow k}(\mv{\theta}_{-k},\mv{v}_k),\cdots, \mv{p}_K\bigr),
\end{align*}
where $\mv{p}_k' = \mv{p}_k$ and, for $j \neq k$,
and $
  \mv{p}_j' = \mv{p}_j + \tfrac{\epsilon}{2}\,\mv{s}_{j\leftarrow k}(\mv{\theta}_{-k},\mv{v}_k).
$

These formulas follow from Hamilton’s equations under $H_k$: $\dot{\mv{\theta}}_k = \mv{v}_k$, $\dot{\mv{\theta}}_j = \mv{0}$ for $j \neq k$, $\dot{\mv{p}}_k = \mv{0}$, and $\dot{\mv{p}}_j = \tfrac{1}{2}\mv{s}_{j\leftarrow k}$ for $j \neq k$. Since $\mv{M}_k$ depends only on $\mv{\theta}_{-k}$ and these blocks remain fixed during the $H_k$-flow, both $\mv{v}_k$ and $\mv{s}_{j\leftarrow k}$ are constant along the flow. Hence
\[
  e^{\epsilon\mathcal{H}_k} = K_k^\epsilon \circ D_k^\epsilon = D_k^\epsilon \circ K_k^\epsilon.
\]

Let $\pi=(\pi_1,\dots,\pi_K)$ be a permutation of $\{1,\dots,K\}$ specifying the order in which the block flows are applied.
\begin{theorem}
\label{thm:multi_symplectic_reversible_hier_leapfrog}
For any permutation $\pi = (\pi_1,\dots,\pi_K)$, Algorithm~\ref{alg:hier_lf_step_pi}
defines the one-step map
\begin{align*}
  \Psi_\epsilon^\pi
  &=
  e^{\frac{\epsilon}{2}\mathcal H_0}
  \circ
  e^{\frac{\epsilon}{2}\mathcal H_{\pi_1}}
  \circ \cdots \circ
  e^{\frac{\epsilon}{2}\mathcal H_{\pi_{K-1}}}
  \circ
  e^{\epsilon\mathcal H_{\pi_K}}
  \circ
  e^{\frac{\epsilon}{2}\mathcal H_{\pi_{K-1}}}
  \circ \cdots \circ
  e^{\frac{\epsilon}{2}\mathcal H_{\pi_1}}
  \circ
  e^{\frac{\epsilon}{2}\mathcal H_0}.
\end{align*}
This map is symmetric (time-reversible), symplectic, and therefore volume-preserving.
\end{theorem}
\begin{proof}
By construction, $\Psi_\epsilon^\pi$ is a palindromic composition of exact Hamiltonian flows. Standard results for symmetric splitting integrators imply that
\[
(\Psi_\epsilon^\pi)^{-1}=\Psi_{-\epsilon}^\pi,
\]
so the method is symmetric. Moreover, each subflow $e^{\epsilon\mathcal{H}_\bullet}$ is symplectic for every $\epsilon>0$, and the composition of symplectic maps is symplectic. Hence $\Psi_\epsilon^\pi$ is symplectic, and therefore volume-preserving; see \citep[Sec.~V.3, pp.~149--154 and Sec.~VI.2, p.~182]{HairerLubichWanner2006GNI}.
\end{proof}

\section{Horseshoe trace plot}
\label{app:horseshoe-trace}

Here we show the trajectories of the mass parameters for the horseshoe prior example.
\begin{figure}
  \centering
  \includegraphics[width=0.32\linewidth]{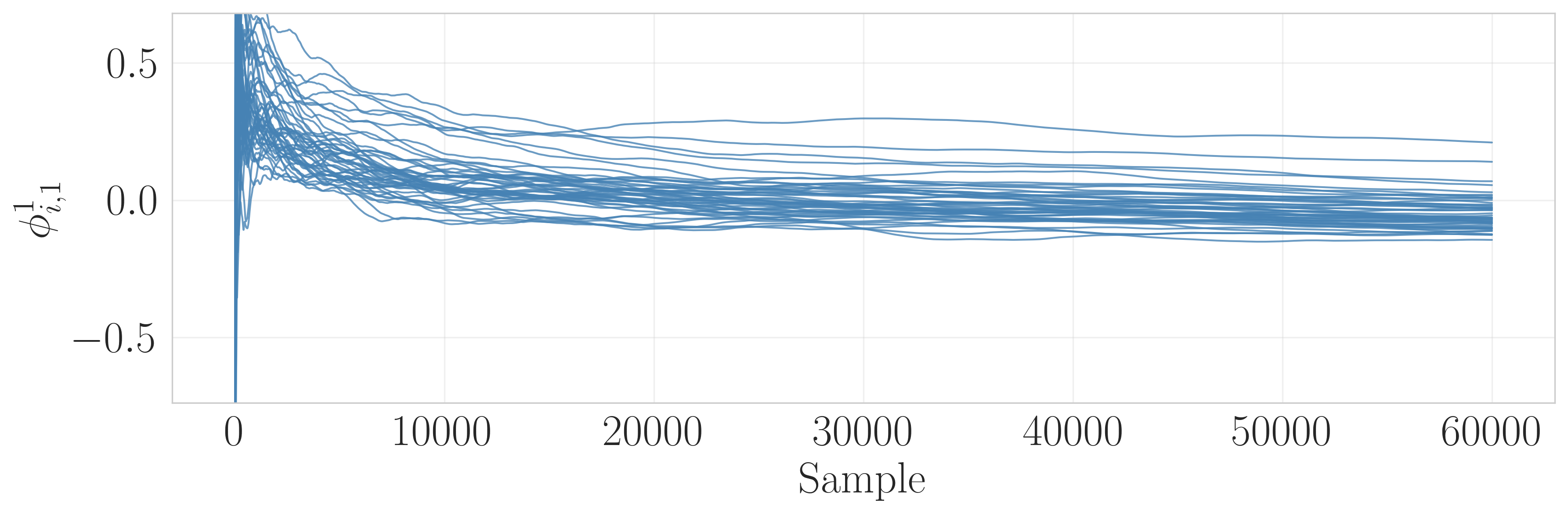}
  \includegraphics[width=0.32\linewidth]{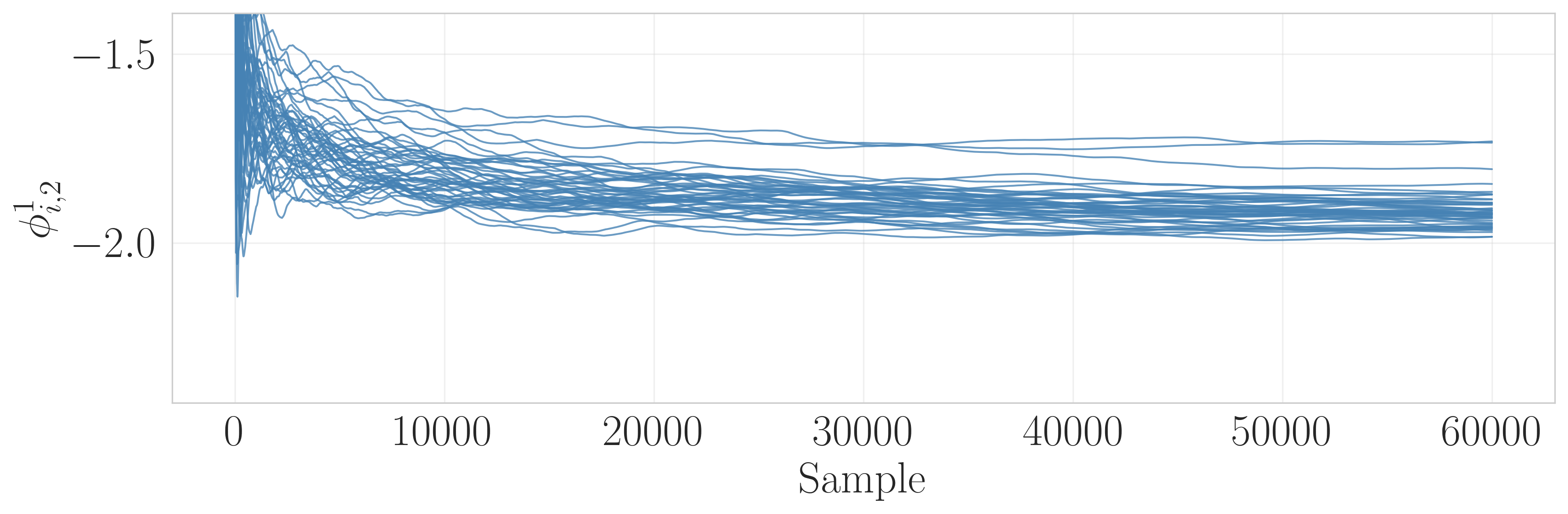}
  \includegraphics[width=0.32\linewidth]{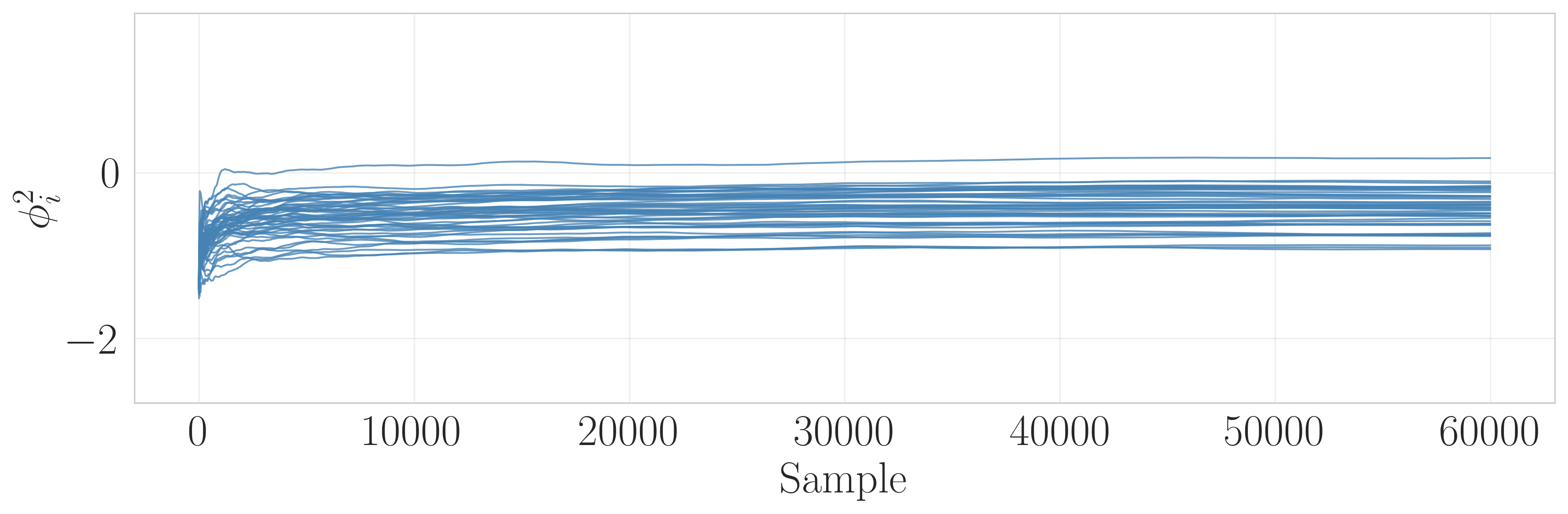} 
  \label{fig:trajectoriesbeta_block_exp}
  \caption{Parameter trajectories for the Block sum-of-exponential-form coefficients}
\end{figure}

\begin{figure}
  \centering
   \includegraphics[width=0.32\linewidth]{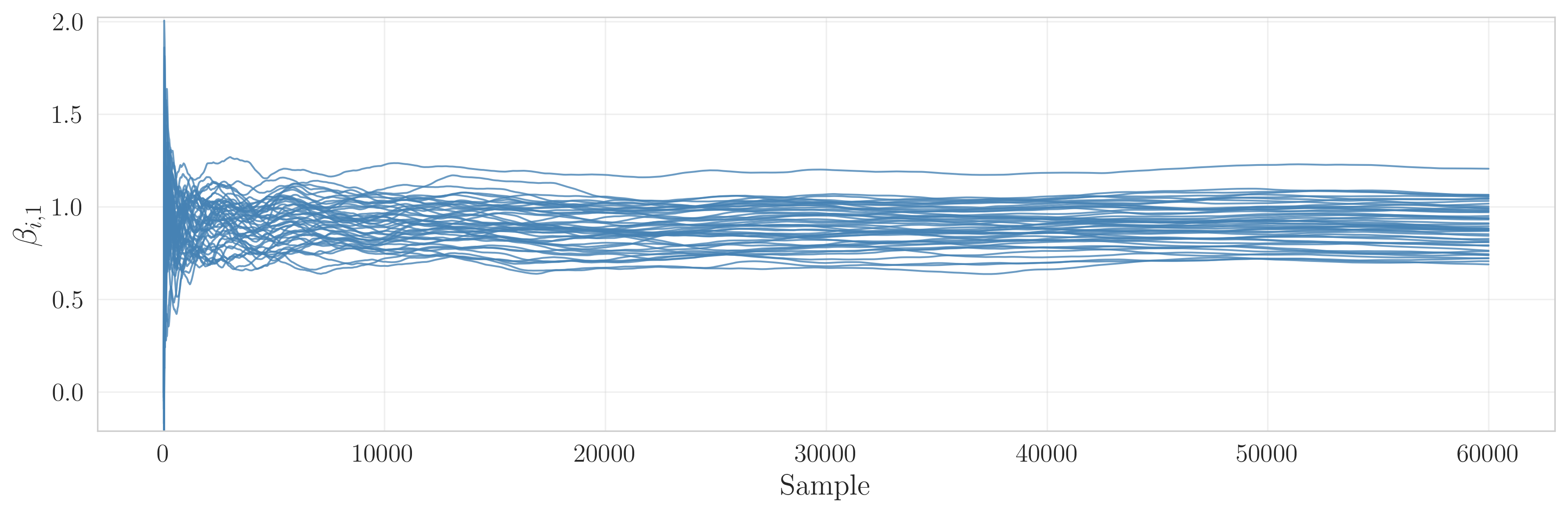} 
  \includegraphics[width=0.32\linewidth]{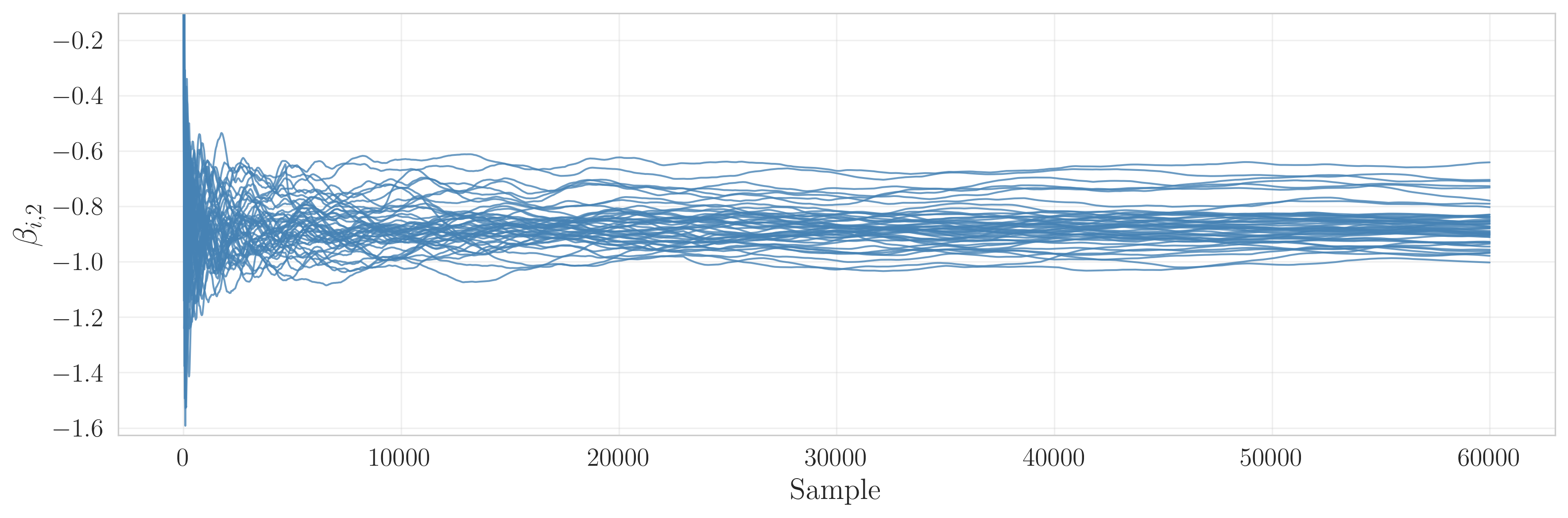} 
 \label{fig:trajectoriesbeta_block_non_additive}
  \caption{Parameter trajectories for the Block exponential coefficients}
\end{figure}

\section{Details of conditional information and Kullback-Leibler minimization}
\label{app:conditional-kl}

Write $\pi = \pi_A\otimes \pi_{B\mid A}$, that is, in terms of its marginal and conditional densities
$$
  \pi(\mv{\theta}) = \pi_A(\mv{\theta}_A) \pi_{B\mid A}(\mv{\theta}_B\mid \mv{\theta}_A).
$$
Note that we may write (the conditional score):
\begin{align*}
\mv{g}_B(\mv{\theta}_B\mid \mv{\theta}_A)
&= \nabla_{\mv{\theta}_B} \log \pi(\mv{\theta}) \\
& = \nabla_{\mv{\theta}_B} \big( \log \pi_A(\mv{\theta})  + \log \pi_{B\mid A}(\mv{\theta}_B\mid \mv{\theta}_A) \big) \\
& = \nabla_{\mv{\theta}_B} \log \pi_{B\mid A}(\mv{\theta}_B\mid \mv{\theta}_A).
\end{align*}

For now, assume that for each fixed $\mv{\theta}_A$, the random vector $\mv{g}_B = \mv{g}_B(\mv{\theta}_B\mid \mv{\theta}_A)$ where $\mv{\theta}_B\sim \pi_{B\mid A}(\uarg\mid \mv{\theta_A})$ has a continuous distribution. This is true under general conditions: it is sufficient that the function $\mv{g}_B(\uarg\mid \mv{\theta}_A)$ is continuously differentiable and that $\det(\nabla_{\mv{\theta}_B} \mv{g}_B(\mv{\theta}_B\mid \mv{\theta}_A)) \neq 0$ almost everywhere. However, it will become clear later, that our \emph{method} is applicable without such a conditional density.

Let $p_{B\mid A}(\uarg\mid \mv{\theta}_A)$ stand for the density of $\mv{g}_B$, and let $\mathcal{Q}_{B\mid A}$ be a family of conditional probability densities $q_{B\mid A}(\mv{\theta}_B \mid \mv{\theta}_A)$ supported on the whole space. Then, for all $q_{B\mid A}\in\mathcal{Q}_{B\mid A}$, we may write
$$
\mathrm{KL}(\pi_A\otimes p_{B\mid A}\;\|\; \pi_A\otimes  q_{B\mid A})
= \E_{\pi_A\otimes p_{B\mid A}}  [\log p_{B\mid A}(\mv{g}_B\mid \mv{\theta}_A) - \log q_{B\mid A}(\mv{g}_B\mid \mv{\theta}_A)].
$$
The first term is independent of $q_{B\mid A}$, so we deduce that
$$
\arg\min_{q_{B\mid A}\in \mathcal{Q}_{B\mid A}}\mathrm{KL}(\pi_A\otimes p_{B\mid A}   \;\|\; \pi_A\otimes q_{B\mid A})
= \arg\min_{q_{B\mid A}\in \mathcal{Q}_{B\mid A}} \E_{\pi_A\otimes p_{B\mid A}}[ -\log q_{B\mid A}(\mv{g}_B\mid \mv{\theta}_A)].
$$

Note that for each fixed $\mv{\theta}_A$, the random vector $\mv{g}_B \sim p_{B\mid A}(\uarg\mid \mv{\theta}_A)$ satisfies by a standard calculation (using Fisher's identity) that $\E[\mv{g}_B]=\mv{0}$ and
$$
  \E[\mv{g}_B \mv{g}_B^T |\mv{\theta}_A] 
  = - \int \nabla_{\mv{\theta}_B} \nabla_{\mv{\theta}_B}^\top \log \pi_{B\mid A}(\mv{\theta}_B\mid \mv{\theta}_A) \pi(\,d\mv{\theta}_B| \mv{\theta}_A)
  = \mathcal{I}_{B\mid A}(\mv{\theta}_A).
$$
This motivates to consider conditional densities $q_{B\mid A}(\uarg \mid \mv{\theta}_A)$ which are zero-mean Gaussian $\mathcal{N}\big(\mv{0},\mv{M}_B(\mv{\theta}_A;\mv{\phi})\big)$, in which case:
$$
   -\log q_{B\mid A}(\mv{g}_B\mid \mv{\theta}_A)
   = \frac{1}{2} \log | \mv{M}_B(\mv{\theta}_A; \mv{\phi}) | + \mv{g}_B \mv{M}_B^{-1}(\mv{\theta}_A; \mv{\phi}) \mv{g}_B,
$$
which is of the form \eqref{eq:per_iter_loss_general}.

\section{Sampling from the trajectory}
\label{sec:nuts-trajectory}
\newcommand{\logpdf}{\log p}
We now describe how we sample from the constructed trajectory in the NUTS algorithm.
We initialize the an iteration by 
$\mv{\theta} \gets \mv{\theta}^{(t)}$ and simulate $\mv{p} \sim \mathcal{N}\!\big(\mv{0}, \mv{M}(\mv{\theta}^{(t)})\big)$.
And describe how we generate a trajectory $\mathcal{T}$ starting from the initial state $z = (\mv{\theta}, \mv{p})$.
The section is based on \cite{Betancourt2017} APPENDIX A.

\subsection{Detailed Trajectory Construction via Multiplicative Expansion}

The construction of the trajectory follows the dynamic implementation strategies detailed in Appendix A of \cite{Betancourt2017}. Let $z = (\mv{\theta}, \mv{p})$ denote a point in phase space. We define a trajectory $\mathcal{T}$ as an ordered sequence of states generated by the symplectic integrator $\Psi_\epsilon$.

The process begins with an initial trajectory $\mathcal{T}_0 = \{ z_0 \}$ containing only the current state of the Markov chain. To ensure the validity of the MCMC sampler, the probability of selecting a specific trajectory $\mathcal{T}$ given an initial state $z$ must satisfy the reversibility condition $\mathbb{T}(\mathcal{T} | z) = \mathbb{T}(\mathcal{T} | z')$ for any $z, z' \in \mathcal{T}$ \cite[App.~A.2.1]{Betancourt2017}. This is achieved by defining the transition distribution as the uniform distribution over the set $\mathfrak{T}_z^L$ of all valid trajectories of length $L$ containing $z$:
\begin{equation}
    \mathbb{T}(\mathcal{T} | z) = \text{Uniform}(\mathfrak{T}_z^L).
\end{equation}
In the dynamic implementation, this set is explored through recursive doubling.
The algorithm used for building a new trajectory starting from $\mathcal{T}_j$ to $\mathcal{T}_{j+1}$ is shown in 
Algorithm \ref{alg:build_tree}.

\begin{algorithm}
\caption{Multiplicative Trajectory Expansion}
\label{alg:build_tree}
\begin{algorithmic}[1]
\Require Current trajectory $\mathcal{T}_j$, step size $\epsilon$
 \State $\mathcal{T}' \gets \emptyset$
 \State $L_j \gets 2^j$
\State  $v \sim \text{Uniform}(\{+1, -1\})$
\If{$v = +1$} \Comment{Forward Expansion}
    \State Let $z_{curr}$ be the last state in $\mathcal{T}_j$
    \For{$k = 1$ to $L_j$}
        \State $z_{curr} \gets \Psi_\epsilon(z_{curr})$ \Comment{Forward Step}
        \State Append $z_{curr}$ to $\mathcal{T}'$
    \EndFor
    \State $\mathcal{T}_{j+1} \gets \mathcal{T}_j \cup \mathcal{T}'$
\Else \Comment{Backward Expansion}
    \State Let $z_{start} = (\mv{\theta}, \mv{p})$ be the first state in $\mathcal{T}_j$
    \State $z_{curr} \gets (\mv{\theta}, -\mv{p})$ \Comment{Negate momentum for reversibility}
    \For{$k = 1$ to $L_j$}
        \State $z_{curr} \gets \Psi_\epsilon(z_{curr})$ \Comment{Forward Step on negated state}
        \State Prepend $(\mv{\theta}_{curr}, -\mv{p}_{curr})$ to $\mathcal{T}'$ \Comment{Store with restored momentum}
    \EndFor
    \State $\mathcal{T}_{j+1} \gets \mathcal{T}' \cup \mathcal{T}_j$
\EndIf
\State \Return $\mathcal{T}_{j+1}$
\end{algorithmic}
\end{algorithm}

\subsection{Termination and Validity Checks}

The termination logic safeguards the Markov 
chain against numerical instability while strictly enforcing the reversibility conditions required for detailed balance.
The first check is ensuring the stability of the symplectic integrator.
 We define the set of Hamiltonian energy deviations for a trajectory segment $\mathcal{T}$ as $\Delta E(\mathcal{T}) = \{ H(\mv{z}) - H(\mv{z}_{init}) \mid \mv{z} \in \mathcal{T} \}$. A trajectory is deemed \textbf{divergent} if the spread of these deviations, $\delta(\mathcal{T}) = \max(\Delta E) - \min(\Delta E)$, exceeds a critical threshold $\Delta_{max}$. If this divergence occurs within the \emph{newly added} segment $\mathcal{T}'$, the integrator has failed locally; the segment is immediately discarded to prevent invalid transitions. Conversely, if the divergence is only detected when considering the \emph{entire} trajectory $\mathcal{T}_j$ (due to the accumulation of global error), the expansion is halted (\texttt{stop}=\texttt{True}), but the trajectory is retained (\texttt{discard}=\texttt{False}) for diagnostic purposes.

If the integration is numerically stable, the algorithm must verify that the expansion preserves the reversibility of the chain. As described in Appendix A.4.1 of \cite{Betancourt2017}, it is insufficient to merely check if the full path has turned back. We must ensure that the new segment does not contain any internal "U-turns" that would have triggered a stop had the algorithm started from a state within that segment. This is verified by checking the \textbf{Generalized No-U-Turn} condition (Algorithm \ref{alg:check_uturn}) on all binary sub-trees within the new segment. If any sub-tree satisfies the U-turn criterion, the new segment violates the reversible tree construction and must be discarded.

Finally, if the segment is both stable and valid, the algorithm checks for termination. If the \emph{full} trajectory $\mathcal{T}_j$ satisfies the No-U-Turn condition, or if the maximum depth is reached, the expansion stops efficiently, accepting the current trajectory as the final proposal.

\begin{algorithm}
\caption{CheckUTurn}
\label{alg:check_uturn}
\begin{algorithmic}[1]
\Require Trajectory segment $\mathcal{T}$ containing states $(\mv{q}_{start}, \mv{p}_{start}) \dots (\mv{q}_{end}, \mv{p}_{end})$
\Ensure Boolean indicating if the trajectory has turned back
\State $\Delta \mv{q} \gets \mv{q}_{end} - \mv{q}_{start}$
\If{$\Delta \mv{q} \cdot \mv{p}_{end} < 0$ \textbf{or} $\Delta \mv{q} \cdot \mv{p}_{start} < 0$}
    \State \Return \texttt{True}
\EndIf
\State \Return \texttt{False}
\end{algorithmic}
\end{algorithm}

\begin{algorithm}
\caption{Check Stopping Conditions}
\label{alg:check_stopping}
\begin{algorithmic}[1]
\Require Current trajectory $\mathcal{T}_j$ (length $2^j$), new segment $\mathcal{T}' \subset \mathcal{T}_j$ (length $2^{j-1}$), threshold $\Delta_{max}$
\Ensure Flags: \texttt{stop}, \texttt{discard}, \texttt{divergent}

\State \textbf{Initialize:} $\texttt{stop}, \texttt{discard}, \texttt{divergent} \gets \text{False}$

\State Calculate $\delta(\mathcal{T}')$ and $\delta(\mathcal{T}_j)$ \Comment{Compute energy spreads}

\If{$\delta(\mathcal{T}') > \Delta_{max}$ }
    \State \Return (\texttt{True}, \texttt{True}, \texttt{True}) \Comment{Local divergence: Discard}
\EndIf

\For{$k = 1$ to $j-1$} \Comment{Check sub-trees in new segment $\mathcal{T}'$}
    \State Let $\{\mathcal{S}_i\}$ be the set of disjoint sub-trees of length $2^k$ inside $\mathcal{T}'$
    \For{each sub-tree $\mathcal{S} \in \{\mathcal{S}_i\}$}
        \If{\textsc{CheckUTurn}($\mathcal{S}$)}
            \State \Return (\texttt{True}, \texttt{True}, \text{False}) \Comment{Invalid: Internal U-turn}
        \EndIf
    \EndFor
\EndFor
\If{$\delta(\mathcal{T}_j) > \Delta_{max}$}
    \State \Return (\texttt{True}, \text{False}, \texttt{True}) \Comment{Global divergence: Stop but retain}
\EndIf
\If{\textsc{CheckUTurn}($\mathcal{T}_j$)} \Comment{Check Global U-Turn}
    \State \texttt{stop} $\gets$ \text{True} 
\EndIf

\If{$j \ge \text{max\_depth}$}
    \State \texttt{stop} $\gets$ \text{True}
\EndIf

\State \Return (\texttt{stop}, \texttt{discard}, \texttt{divergent})
\end{algorithmic}
\end{algorithm}
\subsection{Generalized No-U-Turn Criterion}
\label{sec:generalized_nouturn}

We now present a No-U-Turn criterion that is compatible with a
position-dependent metric $\mv{M}(\mv{q})$ introduced in
\citep{Betancourt2013_rmnuts}.

In the Euclidean case with constant mass matrix $\mv{M}$, the Hamiltonian
flow satisfies $\dot{\mv{q}}(t)=\mv{M}^{-1}\mv{p}(t)$.
If $\mv{M}=\mv{I}$, then
\[
\mv{q}_{end}-\mv{q}_{start}
\;=\;
\int \mv{p}(t)\,dt,
\]
and the standard No-U-Turn condition checks whether the displacement
between the endpoints has become negatively aligned with either endpoint
momentum.

For a general position-dependent metric $\mv{M}(\mv{q})$, Hamiltonian flow gives
\[
\dot{\mv{q}}(t)
\;=\;
\mv{M}(\mv{q}(t))^{-1}\mv{p}(t),
\]
and hence the natural generalized displacement along a segment
$\mathcal{T}$  is
\begin{equation}
\label{eq:rho_segment}
\mv{\rho}(\mathcal{T})
\;:=\;
\int_{\mathcal{T}}
\mv{M}(\mv{q}(t))^{-1}\mv{p}(t)\,dt,
\end{equation}
which in practice is approximated by a discrete sum over the integrator
states in $\mathcal{T}$. This formulation corresponds to the Riemannian generalization of NUTS
proposed by \citet{Betancourt2013_rmnuts}.
In the Euclidean case with $\mv{M}=\mv{I}$ this reduces to
$\mv{\rho}(\mathcal{T})=\mv{q}_{end}-\mv{q}_{start}$.

A trajectory segment is said to have made a U-turn if the generalized
displacement becomes negatively aligned with either endpoint momentum, i.e.
\[
\mv{\rho}(\mathcal{T})^\top \mv{p}_{end} < 0
\quad\text{or}\quad
\mv{\rho}(\mathcal{T})^\top \mv{p}_{start} < 0.
\]
This formulation is consistent with both forward and backward integration,
as reversing the direction of the segment changes the sign of
$\mv{\rho}(\mathcal{T})$.
This can be incorporated into Algorithm~\ref{alg:check_stopping} by replacing each call to \textsc{CheckUTurn} with \textsc{CheckUTurn (Generalized)} in Algorithm~\ref{alg:check_uturn_generalized}.
\begin{algorithm}
\caption{CheckUTurn (Generalized)}
\label{alg:check_uturn_generalized}
\begin{algorithmic}[1]
\Require Trajectory segment $\mathcal{T}$ containing states
$(\mv{q}_{start}, \mv{p}_{start}) \dots (\mv{q}_{end}, \mv{p}_{end})$
\Ensure Boolean indicating if the trajectory has turned back
\State $
\mv{\rho} \gets
\sum_{z\in \mathcal{T}}
\mv{M}(\mv{q})^{-1}\mv{p}\,\Delta t
$
\Comment{Discrete approximation of \eqref{eq:rho_segment}}
\If{$\mv{\rho}^\top \mv{p}_{end} < 0$
    \textbf{or}
    $\mv{\rho}^\top \mv{p}_{start} < 0$}
    \State \Return \texttt{True}
\EndIf
\State \Return \texttt{False}
\end{algorithmic}
\end{algorithm}

\subsection{Retrospective Biased Sampling: Offsets and a Basic Invariance Result}
Let the final (tree-indexed) trajectory be the ordered sequence
\[
\mathcal{T}=\bigl(z_0,z_1,\dots,z_{2^D-1}\bigr),\qquad z_t=(\mv{\theta}_t,\mv{p}_t),
\]
constructed by multiplicative doubling to depth \(D\), and let $H(z_t)= H(\mv{\theta}_t,\mv{p}_t)$  Following
the biased progressive sampling construction of \cite[App.~A.3.2]{Betancourt2017},
we will now build, recursively the sequence of nonnegative sampling weights
\[
w=\bigl(w_0,\dots,w_{2^D-1}\bigr).
\]
The validity (detailed
balance) of the resulting selection rule is inherited from
\cite[App.~A.3.2]{Betancourt2017}.

Initialize the offset and base weights by
\[
o_D := 0,
\qquad
w^{(0)}_t \propto \exp \left(-H(z_t) \right),\quad t=0,1,\dots,2^D-1.
\]
For each level \(\ell=D,D-1,\dots,1\) recurse as follows.  Given the
doubling direction \(v_\ell\in\{+1,-1\}\), update the
offset by $
o_{\ell-1} \;:=\; o_\ell + 2^{\ell-1}\,\mathbf{1}\{v_\ell=-1\}.$
Define the unique size-\(2^\ell\) block containing the initial state \(z_0\) by
$B_\ell \;:=\; \{o_\ell,o_\ell+1,\dots,o_\ell+2^\ell-1\},$ 
which are the indices of the trajectory $\mathcal{T}_l$ in $\mathcal{T}_D$. We then have the associated left/right halves
\[
B_\ell^{\mathrm{L}}:=\{o_\ell,\dots,o_\ell+2^{\ell-1}-1\},\qquad
B_\ell^{\mathrm{R}}:=\{o_\ell+2^{\ell-1},\dots,o_\ell+2^\ell-1\}.
\]
The old half is where the trajectory was build from  (the one containing \(z_0\)) is determined by \(v_\ell\):
\[
(B_{\ell-1},B_\ell^{\mathrm{new}})=
\begin{cases}
(B_\ell^{\mathrm{L}},B_\ell^{\mathrm{R}}), & v_\ell=+1,\\
(B_\ell^{\mathrm{R}},B_\ell^{\mathrm{L}}), & v_\ell=-1,
\end{cases}
\]
note that $B_\ell^{\mathrm{new}}$ is indices of $\mathcal{T}'= \mathcal{T}_{\ell} \setminus \mathcal{T}_{\ell-1}$ for  in $\mathcal{T}_D$.
 With the current weights \(w^{(D-\ell)}\) define the component masses
\[
W_\ell^{\mathrm{old}} := \sum_{t\in B_{\ell-1}} w_t^{(D-\ell)},
\qquad
W_\ell^{\mathrm{new}} := \sum_{t\in B_\ell^{\mathrm{new}}} w_t^{(D-\ell)}.
\]
Set the (component-level) acceptance factor
$
\alpha_\ell \;:=\; \min\!\left(1,\frac{W_\ell^{\mathrm{new}}}{W_\ell^{\mathrm{old}}}\right),
$
and update only the weights inside \(B_\ell\) by
\[
w_t^{(D-\ell+1)} :=
\begin{cases}
c_\ell\,(1-\alpha_\ell)\, w_t^{(D-\ell)}, & t\in B_{\ell-1},\\[4pt]
c_\ell\,\alpha_\ell\, w_t^{(D-\ell)}, & t\in B_\ell^{\mathrm{new}},\\[4pt]
w_t^{(D-\ell)}, & t\notin B_\ell,
\end{cases}
\qquad
c_\ell := \frac{\sum_{t\in B_\ell} w_t^{(D-\ell)}}
{(1-\alpha_\ell)W_\ell^{\mathrm{old}}+\alpha_\ell W_\ell^{\mathrm{new}}}.
\]
After completing the recursion down to \(\ell=1\), define \(w:=w^{(D)}\) and
sample the final state by drawing an index
\[
I \sim \text{Categorical}\!\left(\frac{w_0}{\sum_s w_s},\dots,
\frac{w_{2^D-1}}{\sum_s w_s}\right),
\qquad \text{and returning } z_I\in\mathcal{T}.
\]
The algorithms for the sampling is described in Algorithms \ref{alg:replay_step} and \ref{alg:replay_full}.
\begin{algorithm}
\caption{One retrospective replay update at level $\ell$ (biased merge)}
\label{alg:replay_step}
\begin{algorithmic}[1]
\Require Level $\ell\in\{1,\dots,D\}$, offset $o_\ell$, block $B_\ell$, current weights $w$ (i.e.\ $w^{(D-\ell)}$), directions $v_1,\dots,v_D$
\Ensure Updated tuple $(\ell-1,\,B_{\ell-1},\,o_{\ell-1},\,w,\,v)$ where $w$ is now $w^{(D-\ell+1)}$

\State $B_\ell^{\mathrm{L}} \gets \{o_\ell,\dots,o_\ell+2^{\ell-1}-1\}$
\State $B_\ell^{\mathrm{R}} \gets \{o_\ell+2^{\ell-1},\dots,o_\ell+2^\ell-1\}$

\If{$v_\ell = +1$}
    \State $B_{\ell-1} \gets B_\ell^{\mathrm{L}}$ 
    \State $B_\ell^{\mathrm{new}} \gets B_\ell^{\mathrm{R}}$
\Else
    \State $B_{\ell-1} \gets B_\ell^{\mathrm{R}}$ 
    \State $B_\ell^{\mathrm{new}} \gets B_\ell^{\mathrm{L}}$
\EndIf

\State $W_\ell^{\mathrm{old}} \gets \sum_{t\in B_{\ell-1}} w_t$
\State $W_\ell^{\mathrm{new}} \gets \sum_{t\in B_\ell^{\mathrm{new}}} w_t$
\State $\alpha_\ell \gets \min\!\left(1,\;W_\ell^{\mathrm{new}}/W_\ell^{\mathrm{old}}\right)$
\Comment{if $W_\ell^{\mathrm{old}}=0$, set $\alpha_\ell\gets 1$}

\State $S_\ell \gets \sum_{t\in B_\ell} w_t$
\State $c_\ell \gets
\displaystyle
\frac{S_\ell}{(1-\alpha_\ell)W_\ell^{\mathrm{old}}+\alpha_\ell W_\ell^{\mathrm{new}}}$

\For{$t\in B_{\ell-1}$}
    \State $w_t \gets c_\ell(1-\alpha_\ell)\,w_t$
\EndFor
\For{$t\in B_\ell^{\mathrm{new}}$}
    \State $w_t \gets c_\ell\alpha_\ell\,w_t$
\EndFor

\State $o_{\ell-1} \gets o_\ell + 2^{\ell-1}\mathbf{1}\{v_\ell=-1\}$

\State \Return $(\ell-1,\,B_{\ell-1},\,o_{\ell-1},\,w,\,v)$
\end{algorithmic}
\end{algorithm}

\begin{algorithm}
\caption{Recursive construction of $w=w^{(D)}$ and sampling from $\mathcal{T}$}
\label{alg:replay_full}
\begin{algorithmic}[1]
\Require Trajectory $\mathcal{T}=(z_0,\dots,z_{2^D-1})$, energies $H(z_t)$, directions $v_1,\dots,v_D$
\Ensure Index $I$ and selected state $z_I\in\mathcal{T}$

\State $o_D \gets 0$
\State $B_D \gets \{0,1,\dots,2^D-1\}$
\State $w_t \gets \exp\!\big(-H(z_t)\big)$ for $t=0,\dots,2^D-1$ \Comment{$w^{(0)}$ up to a constant}
\State $\ell \gets D$
\While{$\ell>0$}
    \State $(\ell,\,B_\ell,\,o_\ell,\,w,\,v) \gets \textsc{ReplayStep}(\ell,\,B_\ell,\,o_\ell,\,w,\,v)$
\EndWhile
\State $w \gets w / \sum_s w_s$
\State $I \sim \mathrm{Categorical}(w_0,\dots,w_{2^D-1})$
\State \Return $I,\;z_I$
\end{algorithmic}
\end{algorithm}

\subsection{Step-size adaptation using the last expanded subtree}
\label{app:stepsize-adaptation}

We now want to use the trajectory to adapt the step size $\epsilon$ of the symplectic integrator.
At the final doubling step of the trajectory  we have the new states in
\[
\mathcal{T}' \;:=\; \mathcal{T}_D \setminus \mathcal{T}_{D-1},
\qquad |\mathcal{T}'| = 2^{D-1}.
\]
Define the average acceptance probability over $\mathcal{T}'$ by
\begin{equation*}
\widehat H = 
\frac{1}{|\mathcal{T}'|}
\sum_{z\in\mathcal{T}'}
\exp\!\Big(\min\big(-H(z)+H(z_0),\,0\big)\Big).
\end{equation*}

The statistic $\widehat H$ is used to adapt the step size $\epsilon$ via Polyak-averaged \citep{polyak-juditsky} accelerated Robbins--Monro approximation \citep{kesten1958accelerated}.
Let $\alpha_0 \in (0, 1)$ denote the target average acceptance probability, $n$ the iteration number and $x = \log \epsilon$ the step size on logarithmic scale.
Let $n_0 \in \mathbb{N}$ and $\kappa \in (1/2, 1]$ be tuning parameters (we use $n_0 = 5$ and $\kappa = 0.75$ throughout).
Initializing with $N_{\mathrm{eff}, 0} = 0$ and $\hat x_0 = x_0$, the updates are as follows.

\begin{align*}
    \eta_k^{(0)} & = (n_0+k)^{-\kappa}
    \\
    \eta_k^{(1)} & = (n_0+N_{\mathrm{eff}, k})^{-\kappa}
    \\
    \hat x_{k+1} & = \hat x_k - \eta_k^{(1)} (\alpha_0 - \widehat H_k)
    \\
    x_{k+1} & = (1 - \eta_k^{(0)}) x_k + \eta_k^{(0)} \hat x_{k+1}
    \\
    \sigma_{k+1} & = \operatorname{sgn}(\alpha_0 - \widehat H_k)
    \\
    N_{\mathrm{eff}, k+1} & = N_{\mathrm{eff}, k} + \mathbf{1}[\sigma_{k+1} \sigma_k < 0]
\end{align*}

We remark that it is common to use Nesterov Dual Averaging (NDA) to adapt the step size of NUTS, as in the original paper \cite{HoffmanGelman2014}.
However, we have observed that NDA seems to exhibit strong, slow oscillations of the step size around its targeted value.
Since we do not turn off the adaptation at any point we found this undesirable.
In our experiments, the Polyak-averaged accelerated Robbins--Monro approximation algorithm described above adapts practically as quickly as NDA but results in convergence to rather than oscillations around the target value (results not shown).


\section{Gradient clipping}
\label{sec:clipping-details}
Here we let $\mv{g}$ the gradient of the log-target.
\[
\mv{g}^{\mathrm{clip}}=
\begin{cases}
  \mv{g} & \text{if }\lVert\mv{g}\rVert\le C,\\[4pt]
  C\,\dfrac{\mv{g}}{\lVert\mv{g}\rVert} & \text{otherwise},
\end{cases}
\]
with some threshold $C$. Because the threshold is problem specific, we set the $C$ to the approximate target quantile $\delta\in[0,1]$; we set this to $0.9$ in our experiments. The quantiles are themselves learned by a stochastic approximation rule \citep[cf.][]{duflo,fort-moulines-schreck-vihola}:
\[
C \leftarrow C \exp\!\bigl[-\eta(\mathrm{clip}-\delta)\bigr],
\qquad
\mathrm{clip}=\mathbbm{1}\!\bigl(\lVert\mv{g}\rVert>C\bigr),
\]
so $C$ grows when clipping is too common and shrinks otherwise.

\end{document}